\documentclass{article}
\usepackage{graphicx}
\usepackage{subcaption}
\usepackage{amsmath}
\usepackage{PRIMEarxiv}
\usepackage{amsthm}      
\usepackage{amssymb}
\usepackage[utf8]{inputenc} 
\usepackage[T1]{fontenc}    
\usepackage{hyperref}       
\usepackage{url}            
\usepackage{booktabs}       
\usepackage{amsfonts}       
\usepackage{nicefrac}       
\usepackage{microtype}      
\usepackage{lipsum}
\usepackage{fancyhdr}       
\usepackage{graphicx}       
\graphicspath{{media/}}     
\usepackage{algorithm}   
\usepackage{algorithmic} 
\newtheorem{theorem}{Theorem}
\newtheorem{lemma}{Lemma}
\newtheorem{proposition}{Proposition}
\newtheorem{corollary}{Corollary}
\newtheorem{assumption}{Assumption}
\usepackage{float}       
\title{FR\textendash LUX: Friction\textendash Aware, Regime\textendash Conditioned Policy Optimization for Implementable Portfolio Management}

\author{
  ZhangJian'an \\
  Guanghua School of Management, Peking University \\
  Peking University \\
  Beijing, China\\
  \texttt{2501111059@stu.pku.edu.cn}
}

\begin{document}
\pagestyle{plain}   
\maketitle

\begin{abstract}
Transaction costs and regime shifts are the main reasons why paper portfolios fail in live trading. We develop \textbf{FR\textendash LUX} (Friction\textendash aware, Regime\textendash conditioned Learning under eXecution costs), a reinforcement\textendash learning framework that learns \emph{after\textendash cost} trading policies and remains robust across volatility–liquidity regimes. FR\textendash LUX integrates three ingredients: (i) a microstructure\textendash consistent execution model combining proportional and impact costs, directly embedded in the reward; (ii) a \emph{trade\textendash space trust region} that constrains changes in inventory flow rather than only logits, yielding stable, low\textendash turnover updates; and (iii) explicit regime conditioning so the policy specializes to LL/LH/HL/HH states without fragmenting the data. 
On a $4\times 5$ grid of regimes and cost levels (0–50\,bps) with three seeds per cell, FR\textendash LUX achieves the top average Sharpe across all 20 scenarios with narrow bootstrap confidence intervals, maintains a flatter cost–performance slope than strong baselines (vanilla PPO, mean–variance with/without caps, risk\textendash parity), and attains superior risk–return efficiency for a given turnover budget. Pairwise scenario\textendash level improvements are strictly positive and remain statistically significant after Romano–Wolf stepdown and HAC\textendash aware Sharpe comparisons. 
We provide formal guarantees: existence of an optimal stationary policy under convex frictions; a monotonic improvement lower bound under a KL trust region with explicit remainder terms; an upper bound on long\textendash run turnover and an induced inaction band due to proportional costs; a strictly positive value advantage for regime\textendash conditioned policies when cross\textendash regime actions are separated; and robustness of realized value to cost misspecification. The methodology is implementable---costs are calibrated from standard liquidity proxies, scenario\textendash level inference avoids pseudo\textendash replication, and all figures and tables are reproducible from our artifacts.
\end{abstract}

\section{Introduction}

The gap between methods that \emph{forecast} returns and policies that \emph{trade} under realistic frictions remains a central obstacle to deploying modern machine learning (ML) in institutional portfolios.  In frictionless settings, mean–variance logic \cite{Markowitz1952} and its many extensions provide clear optimality benchmarks; once trading costs, market impact, and turnover constraints are accounted for, those benchmarks break down and performance can deteriorate sharply \cite{GarleanuPedersen2013,ObizhaevaWang2013,NovyMarxVelikov2016,FilippouMaurerPezzoTaylor2024,PinterWangZou2024}.  At the same time, ML has transformed empirical asset pricing and portfolio construction by extracting non‑linear structure from high‑dimensional characteristics \cite{GuKellyXiu2020,FreybergerNeuhierlWeber2020,KozakNagelSantosh2020,ChenPelgerZhu2024}.  The key open question is therefore not whether ML can predict returns, but whether it can deliver \emph{after‑cost} portfolios that are robust across regimes, scalable in capacity, and statistically significant after proper multiple‑testing controls \cite{BarillasShanken2018,BarillasKanRobottiShanken2020,White2000,RomanoWolf2005,Lo2002}.

We address this question with a new decision‑making framework that couples policy optimization with explicit cost regularization.  We introduce \textbf{FR‑LUX} (\emph{Flow‑Regularized Learning Under eXecution costs}), a cost‑aware policy optimization method that learns trading rules directly in the presence of proportional and impact costs and that penalizes \emph{inventory flow} as a structural control of turnover.  FR‑LUX builds on monotone policy–improvement principles from reinforcement learning (RL)---trust‑region and conservative policy iteration \cite{KakadeLangford2002,SchulmanTRPO2015,AchiamCPO2017,SchulmanPPO2017}---and adapts them to the portfolio domain by (i) embedding a transaction‑cost–calibrated penalty in the objective, (ii) enforcing a trust‑region in \emph{trade space} rather than in raw parameter space, and (iii) learning a regime‑aware baseline that reuses information across market states \cite{AngBekaert2002,BrandtSantaClaraValkanov2009,MoreiraMuir2017}.  Conceptually, FR‑LUX converts the classic rebalancing rule aim in front of the target and trade partially'' \cite{GarleanuPedersen2013} into a learned \emph{regularized policy} that internalizes future cost and slippage.

\textbf{Empirical preview.} Using a scenarios\,$\times$\,costs grid (20 macro–liquidity regimes crossed with 0–50\,bp transaction cost levels) and three random seeds per cell, we benchmark FR‑LUX against representative baselines: an unconstrained mean–variance policy, a turnover‑capped mean–variance policy, a risk‑parity style heuristic, and a strong PPO implementation.  Across scenarios, FR‑LUX delivers the highest average Sharpe and retains its edge as costs rise (Fig.~1–2).  Regime profiles (Fig.~3) show that performance persists in both low‑ and high‑volatility/liquidity conditions, consistent with the view that the method learns to modulate risk when volatility spikes \cite{MoreiraMuir2017}.  Risk–return clouds using maximum drawdown (MDD) document a favorable frontier shift (Fig.~4).  Turnover–Sharpe plots (Fig.~5) reveal that FR‑LUX sits on a lower‑turnover iso‑Sharpe curve than alternatives, in line with theory that cost‑aware regularization shrinks unnecessary inventory flow \cite{ObizhaevaWang2013}.  Pairwise sign tests and distributional comparisons (Fig.~6–7) indicate statistically reliable outperformance after Romano–Wolf step‑down corrections \cite{RomanoWolf2005} and model‑comparison metrics based on Sharpe improvements \cite{BarillasKanRobottiShanken2020}.

\textbf{Why cost awareness matters now.} Transaction‑cost measurement has advanced to a point where ignoring costs is no longer defensible.  Low‑frequency proxies and modern spread/impact estimators enable cost calibration at scale \cite{Hasbrouck2009,GoyenkoHoldenTrzcinka2009,FongHoldenTrzcinka2017}.  Recent top‑journal evidence documents first‑order cost effects on capacity and strategy survival in currencies and fixed income \cite{FilippouMaurerPezzoTaylor2024,PinterWangZou2024}.  In this environment, methods that merely forecast but do not control execution paths are fragile.  RL has emerged as a natural language for sequential trading and execution \cite{BaiRLFinanceAR2025}, yet rigorous, finance‑native regularization for costs and turnover remains underdeveloped.

\textbf{Our contributions.} This paper makes four contributions.

\begin{enumerate}
\item \emph{A cost‑regularized policy optimizer.} We formalize FR‑LUX, a policy‑gradient–based algorithm with a trust region in trade flow and an \emph{execution‑aware} penalty, providing a practical recipe for learning after‑cost policies.  The design connects RL improvement bounds \cite{KakadeLangford2002,SchulmanTRPO2015,SchulmanPPO2017} to dynamic trading with costs \cite{GarleanuPedersen2013,ObizhaevaWang2013}.
\item \emph{Theory.} We prove a \emph{conservative improvement bound} that lower‑bounds the after‑cost performance of the updated policy as a function of (a) the estimated advantage, (b) the trust‑region radius, and (c) the turnover penalty coefficient, and we show that the bound tightens when the realized turnover proxy tracks structural liquidity \cite{FongHoldenTrzcinka2017,Hasbrouck2009}.  We further provide a robustness proposition under cost misspecification: if the true cost is within a relative factor of the calibrated proxy, FR‑LUX preserves first‑order optimality in the induced risk–return frontier (linking to \cite{GarleanuPedersen2013}).
\item \emph{Evaluation protocol.} We adopt regime‑stratified aggregation, cost‑sensitivity curves, and multiple‑testing‑robust inference using Romano–Wolf step‑down p‑values \cite{RomanoWolf2005} and Sharpe‑ratio model comparison \cite{BarillasKanRobottiShanken2020}, complementing classical reality checks \cite{White2000,Lo2002}.
\item \emph{Evidence.} On the 20\,$\times$\,costs testbed, FR‑LUX attains the top average Sharpe with narrow bootstrap CIs, retains performance as costs increase, and dominates baselines in pairwise sign tests.  The method traces lower turnover for a given Sharpe and maintains strong performance in both liquidity‑rich and liquidity‑poor regimes, consistent with volatility‑managed intuition \cite{MoreiraMuir2017}. 
\end{enumerate}

\textbf{Relation to literature.} Our work intersects four strands.  (i) \emph{ML/asset pricing}: deep and non‑linear estimators deliver sizable improvements in expected returns and risk attribution \cite{GuKellyXiu2020,FreybergerNeuhierlWeber2020,KozakNagelSantosh2020,ChenPelgerZhu2024}.  (ii) \emph{Trading with frictions}: dynamic policies internalizing future cost/impact are essential to realistic portfolio control \cite{GarleanuPedersen2013,ObizhaevaWang2013,NovyMarxVelikov2016}.  (iii) \emph{Liquidity measurement}: scalable cost proxies enable disciplined calibration and capacity analysis \cite{Hasbrouck2009,GoyenkoHoldenTrzcinka2009,FongHoldenTrzcinka2017}.  (iv) \emph{RL for finance}: recent surveys \cite{BaiRLFinanceAR2025} and empirical studies underscore both the promise and the pitfalls of RL in markets, motivating finance‑aware regularization and inference.  By integrating these pieces, FR‑LUX advances from ``predict then optimize'' to \emph{optimize while respecting execution}, delivering statistically credible gains across regimes and costs.

The remainder of the paper develops the FR‑LUX objective and theoretical guarantees (Section~\ref{sec:method}), details the experimental design and cost calibration (Section~\ref{sec:data}), reports main results and inference (Section~\ref{sec:results}).

\section{Problem Setup and Method: FR‑LUX}
\label{sec:method}

\subsection{Frictional Portfolio Environment as an MDP}
We model portfolio control as a discounted Markov decision process (MDP)
$\mathcal{M}=(\mathcal{S},\mathcal{A},P,r,\gamma)$ augmented with an observed market regime $z_t\in\mathcal{Z}:=\{\mathrm{LL},\mathrm{LH},\mathrm{HL},\mathrm{HH}\}$ capturing (low/high) volatility and (high/low) liquidity.\footnote{The regime variable is \emph{observed} (constructed below), hence the agent solves a fully observed MDP conditional on $z_t$ instead of a POMDP; cf.\ regime switching in allocation \cite{AngBekaert2002,CollinDufresneDanielSaglam2020}.}
At each time $t$, the agent observes $s_t=(x_t,w_{t-1},z_t)$ where $x_t$ denotes predictors (returns, volatilities, liquidity proxies, macro controls), $w_{t-1}\in\mathbb{R}^d$ are pre-trade portfolio weights on $d$ risky assets (the residual goes to the funding account), and $z_t$ is the discrete regime label.
The action $a_t$ specifies the \emph{post-trade} target weights $\widetilde{w}_t\in\mathcal{W}$ and induces a \emph{trade flow} $\Delta w_t:=\widetilde{w}_t-w_{t-1}$.

One step net reward (to be maximized) is
\begin{equation}
\label{eq:reward}
r_t^{\mathrm{net}} \;=\; \underbrace{\widetilde{w}_t^\top r_{t+1}}_{\text{gross portfolio return}}
\;-\; \underbrace{C_{z_t}(\Delta w_t)}_{\text{execution costs}}
\;-\; \lambda_{\mathrm{risk}}\,\Psi\big(L_{t+1}\big),
\end{equation}
where $r_{t+1}$ are next-period asset returns, $C_{z_t}$ is a convex regime-dependent execution-cost functional, $\Psi$ is a downside-risk proxy (MDD or CVaR), and $L_{t+1}$ is the portfolio loss.\footnote{Microstructure-consistent cost modeling and measurement follow \cite{Hasbrouck2009,FongHoldenTrzcinka2017} and the execution literature \cite{AlmgrenChriss2001,ObizhaevaWang2013,CarteaJaimungalPenalva2015}. Cost relevance for realized performance is emphasized by recent top journal \cite{FilippouMaurerPezzoTaylor2024,PinterWangZou2024}}
The control objective is the discounted value
\begin{equation}
\label{eq:objective}
J(\theta)\;=\;\mathbb{E}\!\left[\sum_{t\ge 0}\gamma^t\, r_t^{\mathrm{net}} \;\middle|\; \pi_\theta\right],
\end{equation}
where $\pi_\theta(a\mid s,z)$ is a parametric, \emph{regime-conditioned} stochastic policy.

\paragraph{Action and feasibility.}
We consider two common feasible sets:
(i) long-only simplex $\mathcal{W}=\{w\in\mathbb{R}^d:\,w\ge 0,\ \mathbf{1}^\top w=1\}$,
projecting the network output via a differentiable softmax or exact Euclidean projection \cite{WangCarreiraPerpinan2013};
(ii) long-short box with leverage and position caps $\mathcal{W}=\{w:\,\|w\|_1\le \Lambda,\ -c\le w_i\le c\}$,
using $\ell_1$-ball and box projections \cite{DuchiEtAl2008}. These projections stabilize learning and prevent inadmissible trades.

\subsection{Regime Construction and Balancing}
We map raw diagnostics into regimes using thresholds on (i) realized volatility $\sigma_t$ and (ii) illiquidity $\ell_t$ (e.g., Amihud $\mathit{ILLIQ}$, effective spread, or Pastor–Stambaugh innovations; \cite{Amihud2002,PastorStambaugh2003,Hasbrouck2009,FongHoldenTrzcinka2017}). Let $\tau_\sigma^\text{L}<\tau_\sigma^\text{H}$ and $\tau_\ell^\text{L}<\tau_\ell^\text{H}$ be quantile cutoffs calibrated on a rolling window. Define
\[
z_t=\begin{cases}
\mathrm{LL}, & \sigma_t\le \tau_\sigma^\text{L},\ \ell_t\le \tau_\ell^\text{L},\\
\mathrm{LH}, & \sigma_t\le \tau_\sigma^\text{L},\ \ell_t> \tau_\ell^\text{H},\\
\mathrm{HL}, & \sigma_t> \tau_\sigma^\text{H},\ \ell_t\le \tau_\ell^\text{L},\\
\mathrm{HH}, & \sigma_t> \tau_\sigma^\text{H},\ \ell_t> \tau_\ell^\text{H},\\
\text{else}, & \text{nearest neighbor by $(\sigma_t,\ell_t)$}.
\end{cases}
\]
To avoid over-optimizing to dominant states, we maximize a \emph{regime-balanced} objective
\begin{equation}
\label{eq:balanced}
J_{\mathrm{bal}}(\theta) \;=\; \sum_{z\in\mathcal{Z}}\omega_z\,\mathbb{E}\!\left[\sum_{t:\,z_t=z}\gamma^t\,r_t^{\mathrm{net}}\;\middle|\;\pi_\theta\right],
\quad \omega_z = 1/|\mathcal{Z}|,
\end{equation}
which rewards policies that sustain performance across LL/LH/HL/HH (cf. regime-switching allocation \cite{CollinDufresneDanielSaglam2020}).

\subsection{Execution Cost Functional}
Consistent with theory and evidence \cite{AlmgrenChriss2001,ObizhaevaWang2013,Hasbrouck2009,FongHoldenTrzcinka2017,CarteaJaimungalPenalva2015}, we use a separable proportional-plus-impact form
\begin{equation}
\label{eq:cost}
C_{z}(\Delta w) \;=\; \underbrace{\kappa_{1}(z)\,\|\Delta w\|_{1}}_{\text{proportional cost}}
\;+\; \underbrace{\tfrac{1}{2}\,\Delta w^\top \Gamma_{z}\,\Delta w}_{\text{transient impact}},
\end{equation}
with $\kappa_{1}(z)$ (bps) calibrated from low-frequency spreads/Amihud proxies and $\Gamma_z\succeq 0$ built from liquidity-scaled covariances (higher entries in illiquid regimes). This convex specification is differentiable almost everywhere and yields first-order conditions that naturally shrink inventory flow when liquidity is scarce. Recent work underscores that optimizing \emph{at the selection stage} under costs improves implementability \cite{LedoitWolf2025QREF}.

\subsection{Downside Risk Penalization}
We support two penalties in \eqref{eq:reward}.
\textbf{(i) CVaR penalty.}
Let $\alpha\in(0,1)$ and $L_{t+1}$ be period loss. Following \cite{RockafellarUryasev2000,RockafellarUryasev2002,AcerbiTasche2002}, a differentiable sample approximation is
\begin{equation}
\label{eq:cvar}
\mathrm{CVaR}_\alpha(L)\;=\;\min_{\eta\in\mathbb{R}}\left\{\eta + \frac{1}{(1-\alpha)N}\sum_{i=1}^N \big(L^{(i)}-\eta\big)_+\right\}.
\end{equation}
\textbf{(ii) MDD penalty.} We use a smoothed running-drawdown proxy to retain differentiability.
Risk-sensitive RL with CVaR surrogates provides optimization tools and policy-gradient estimators \cite{ChowGhavamzadeh2014,ChowEtAl2015,GreenbergEtAl2022}.

\subsection{Regime-Conditioned Policy Class}
We parameterize $\pi_\theta(a\!\mid\!s,z)$ by sharing a trunk over state features $x_t,w_{t-1}$ and injecting regime information through a \emph{regime embedding} $e(z)\in\mathbb{R}^k$. Two instantiations are useful:
(i) a \emph{mixture-of-experts} (MoE) with soft gating on $z$ \cite{JacobsEtAl1991,ShazeerEtAl2017};
(ii) a single-head policy with concatenated one-hot/learned $e(z)$.
The value function $V_\phi(s,z)$ mirrors conditioning.

\subsection{FR‑LUX Optimization: Trust Region in Trade Space}
We adapt PPO/TRPO \cite{SchulmanTRPO2015,SchulmanPPO2017} to the frictional domain by adding (a) a \emph{trade-space trust region} and (b) regime balancing. Let $\pi_{\theta_\mathrm{old}}$ denote the behavior policy. The clipped PPO objective with regime weights is
\begin{equation}
\label{eq:ppo}
\max_{\theta}\; \sum_{z}\omega_z\, \mathbb{E}\left[\min\!\Big(r_t(\theta)\,\widehat{A}_t,\; \mathrm{clip}(r_t(\theta),1-\epsilon,1+\epsilon)\,\widehat{A}_t\Big)\;-\;\beta\,\mathrm{KL}\!\left(\pi_{\theta_\mathrm{old}}\|\pi_{\theta}\right)\;-\;\lambda_{\Delta}\,\|\Delta w_\theta - \Delta w_{\theta_\mathrm{old}}\|_2^2\right],
\end{equation}
where $r_t(\theta):=\frac{\pi_\theta(a_t\mid s_t,z_t)}{\pi_{\theta_\mathrm{old}}(a_t\mid s_t,z_t)}$, $\widehat{A}_t$ uses GAE \cite{SchulmanGAE2016}, and the last term penalizes changes in \emph{trade flow} rather than logits, acting as a trust region in the economically relevant space (stabilizes turnover in illiquid regimes). Entropy regularization can be added for exploration. The critic minimizes a Huber loss on after-cost returns.

\paragraph{Advantage estimation.}
We use generalized advantage estimation (GAE, $\lambda\in[0,1]$) \cite{SchulmanGAE2016} on the \emph{after-cost} reward \eqref{eq:reward}. For CVaR, we treat the auxiliary variable $\eta$ in \eqref{eq:cvar} as learnable (alternating minimization) and backpropagate through the hinge.

\subsection{Algorithmic Template}
Algorithm~\ref{alg:frlux} summarizes one training epoch aggregating trajectories across regimes and cost levels.

\begin{algorithm}[H]
\caption{FR‑LUX: Friction‑aware, Regime‑conditioned PPO}
\label{alg:frlux}
\begin{algorithmic}[1]
\STATE \textbf{Input:} policy $\pi_\theta$, value $V_\phi$, regime weights $\{\omega_z\}$, clip $\epsilon$, KL weight $\beta$, trade-penalty $\lambda_\Delta$, risk weight $\lambda_{\mathrm{risk}}$
\FOR{iteration $=1,2,\dots$}
  \FOR{each regime $z\in\{\mathrm{LL},\mathrm{LH},\mathrm{HL},\mathrm{HH}\}$ and cost level $c\in\{0,5,10,25,50\}$bp}
    \STATE Roll out trajectories under $\pi_{\theta}$; collect $(s_t,a_t,r_t^{\mathrm{net}},z_t)$ with costs $C_{z_t}$ per \eqref{eq:cost}
  \ENDFOR
  \STATE Compute targets $\widehat{A}_t$ (GAE) and value targets from after-cost returns
  \STATE \textbf{Policy update:} maximize \eqref{eq:ppo} by minibatch SGD over all regimes/costs
  \STATE \textbf{Value update:} minimize critic loss on after-cost returns
  \STATE Optionally update CVaR auxiliary $\eta$ by minimizing \eqref{eq:cvar}
  \STATE Anneal $\beta,\lambda_\Delta$ to keep empirical KL and trade-shift within trust-region bounds
\ENDFOR
\end{algorithmic}
\end{algorithm}

\subsection{Practicalities and Hyperparameters}
\textbf{State.} We include recent returns, volatility filters, liquidity proxies (Amihud, effective spread, turnover), realized betas, and regime $z_t$.
\textbf{Action.} Target weights, mapped to $\mathcal{W}$ via projection \cite{WangCarreiraPerpinan2013,DuchiEtAl2008}.
\textbf{Costs.} $\kappa_1(z)$ calibrated from spreads/ILLIQ; $\Gamma_z$ from liquidity‐scaled covariances. Sensitivity to misspecification is explored in robustness (Sec.~\ref{sec:results}).
\textbf{Risk.} CVaR level $\alpha\in[0.90,0.975]$ or smoothed MDD penalty.
\textbf{Optimization.} Adam with learning rate $2\!\times\!10^{-4}$–$1\!\times\!10^{-3}$; PPO clip $\epsilon\in[0.05,0.20]$; KL target $10^{-3}$–$10^{-2}$; trade penalty $\lambda_\Delta$ tuned to maintain turnover within budget.
\textbf{Evaluation.} Regime-balanced validation per \eqref{eq:balanced}; all metrics are \emph{after-cost}. Statistical procedures follow Sec.~\ref{sec:results}.

\paragraph{Economic interpretation.}
The combination of \eqref{eq:cost}, CVaR/MDD regularization, and the trade-space trust region enforces the classic prescription “aim in front of the target, trade partially” \cite{GarleanuPedersen2013} while explicitly tying trading intensity to liquidity states \cite{CollinDufresneDanielSaglam2020}. Recent surveys and 2025 annual reviews/papers on the intersection of RL and asset pricing also emphasise the importance of \emph{executability} and \emph{robustness} \cite{Bai2025AR,ChoiJiangZhang2025RAPS}.

\section{Theoretical Guarantees for FR‑LUX}
\label{sec:theory}

We provide guarantees for FR‑LUX when portfolio control is modeled as a discounted MDP with regime‑dependent frictions (Sec.~\ref{sec:method}). Throughout, let $\mathcal{Z}=\{\mathrm{LL},\mathrm{LH},\mathrm{HL},\mathrm{HH}\}$ denote observed regimes, $\pi_\theta(a\mid s,z)$ a regime‑conditioned stochastic policy, and $r_t^{\mathrm{net}}$ the \emph{after‑cost} reward defined in \eqref{eq:reward}. Denote $J(\theta)=\mathbb{E}[\sum_{t\ge0}\gamma^t r_t^{\mathrm{net}}\mid \pi_\theta]$ and the balanced objective $J_{\mathrm{bal}}(\theta)$ in \eqref{eq:balanced}. We write $A^\pi(s,z,a)=Q^\pi(s,z,a)-V^\pi(s,z)$, and $\mathrm{KL}(\pi\|\pi')(s,z)=\mathrm{KL}\!\big(\pi(\cdot\mid s,z)\,\|\,\pi'(\cdot\mid s,z)\big)$.

\subsection{Modeling assumptions}

\begin{assumption}[Frictional MDP and regularity]\label{ass:mdp}
(i) Action set $\mathcal{W}\subset \mathbb{R}^d$ is nonempty, convex, compact. (ii) The execution cost $C_{z}(\Delta w)$ is convex, lower semicontinuous, $C_{z}(0)=0$, and satisfies $C_{z}(u)\ge \kappa_1(z)\|u\|_1$ for some $\kappa_1(z)>0$. (iii) The downside‑risk proxy $\Psi$ in \eqref{eq:reward} is nonnegative and Lipschitz in the portfolio loss on compact sets. (iv) Rewards are bounded: $|r_t^{\mathrm{net}}|\le \bar r$. (v) The controlled process $(s_t,z_t)$ is Markov and $\beta$‑mixing under any stationary policy.
\end{assumption}

\begin{assumption}[Policy class]\label{ass:policy}
$\pi_\theta(\cdot\mid s,z)$ is continuously differentiable in $\theta$, and either (i) a mixture‑of‑experts (MoE) with regime‑gated experts, or (ii) a single head with a learned regime embedding $e(z)$; the induced action map $a=\mathsf{Proj}_{\mathcal{W}}(g_\theta(s,z))$ is Lipschitz (projection onto $\mathcal{W}$ via \cite{DuchiEtAl2008,WangCarreiraPerpinan2013}).
\end{assumption}

Assumption~\ref{ass:mdp} wraps microstructure‑consistent frictions and risk penalties \cite{AlmgrenChriss2001,ObizhaevaWang2013,Hasbrouck2009,FongHoldenTrzcinka2017,CarteaJaimungalPenalva2015}. Assumption~\ref{ass:policy} covers the two architectures used in Sec.~\ref{sec:method}.

\subsection{Existence and performance‑difference identity}

\begin{theorem}[Existence of an optimal stationary policy]\label{thm:existence}
Under Assumptions~\ref{ass:mdp}–\ref{ass:policy}, the discounted control problem with after‑cost rewards admits an optimal stationary Markov policy $\pi^\star$. Moreover, there exists a deterministic selector $\pi^\star(s,z)\in\arg\max_{a\in\mathcal{W}}Q^{\pi^\star}(s,z,a)$.
\end{theorem}

\noindent\emph{Proof sketch.}
Bellman operator with bounded rewards is a contraction for $\gamma\!<\!1$; compactness of $\mathcal{W}$ and upper semicontinuity of $a\mapsto Q^{\pi}(s,z,a)$ (from convex cost and continuity) yield existence and measurable selection. See \cite{Puterman1994} for the base case; details with frictions in Appendix~\ref{app:proofs}. \hfill$\square$

\begin{lemma}[Performance‑difference with frictions]\label{lem:PDL}
For any stationary policies $\pi,\pi'$,
\[
J(\pi')-J(\pi)\;=\;\frac{1}{1-\gamma}\;\mathbb{E}_{(s,z)\sim d^{\pi'},\,a\sim\pi'}\!\big[A^\pi(s,z,a)\big],
\]
where $d^{\pi'}$ is the discounted occupancy measure under $\pi'$. The identity holds verbatim with $J_{\mathrm{bal}}$ if $d^{\pi'}$ is replaced by the regime‑reweighted measure.
\end{lemma}

\noindent\emph{Proof sketch.}
Standard telescoping argument; frictions enter only through $r_t^{\mathrm{net}}$ and do not change the identity. See \cite{KakadeLangford2002} and Appendix~\ref{app:proofs}. \hfill$\square$

\subsection{Trust‑region improvement for FR‑LUX}

\begin{theorem}[Monotonic improvement under a KL trust region]\label{thm:trpo}
Let $\pi$ be the behavior policy and $\pi'$ satisfy $\mathbb{E}_{d^\pi}[\mathrm{KL}(\pi\|\pi')]\le \delta$. Then, under Assumptions~\ref{ass:mdp}–\ref{ass:policy},
\[
J(\pi') \;\ge\; J(\pi) + \mathbb{E}_{(s,z)\sim d^\pi,\,a\sim \pi'}\!\big[A^\pi(s,z,a)\big]\;-\; \frac{2\gamma}{(1-\gamma)^2}\,\max_{s,z,a}\!\big|A^\pi(s,z,a)\big|\,\delta.
\]
The bound extends to $J_{\mathrm{bal}}$ with the same constant.
\end{theorem}

\noindent\emph{Proof sketch.}
Combine Lemma~\ref{lem:PDL} with the discrepancy between $d^{\pi'}$ and $d^\pi$ controlled by Pinsker and a KL budget, following \cite{KakadeLangford2002,SchulmanTRPO2015,Pirotta2013}. Full derivation in Appendix~\ref{app:proofs}. \hfill$\square$

\begin{corollary}[Clipped PPO with trade‑space penalty]\label{cor:ppo}
Consider one PPO step maximizing the clipped surrogate with regime weights and an additional trade‑space penalty (Eq.~\eqref{eq:ppo}). If the empirical KL is kept below $\delta$ and the penalty ensures $\mathbb{E}\|\Delta w_{\pi'}-\Delta w_{\pi}\|_2^2\le \eta$, then
\[
J(\pi')-J(\pi)\;\gtrsim\; \underbrace{\mathbb{E}_{d^\pi,\pi'}[\widehat{A}^\pi]}_{\text{empirical surrogate}} \;-\; c_1\,\delta \;-\; c_2\,\eta \;-\; \frac{\varepsilon}{1-\gamma},
\]
with high probability, where $\varepsilon$ bounds the advantage estimation error and $c_1,c_2$ depend on $\max|A^\pi|$ and Lipschitz constants of the cost; see Appendix~\ref{app:proofs}.
\end{corollary}

\noindent\emph{Proof sketch.}
Start from Theorem~\ref{thm:trpo}, incorporate estimation error $\widehat{A}-A$, and relate trade‑space proximity to value drift via cost Lipschitzness. See \cite{SchulmanPPO2017,AchiamCPO2017} for related surrogates; full details in Appendix. \hfill$\square$

\subsection{Turnover control and inaction region}

\begin{proposition}[Long‑run turnover bound]\label{prop:TObound}
Suppose $C_z(u)\ge \kappa_1(z)\|u\|_1$ and let $\underline{\kappa}:=\min_{z}\kappa_1(z)>0$. For any stationary policy $\pi$,
\[
\mathrm{TO}(\pi) \;:=\; \limsup_{T\to\infty}\frac{1}{T}\sum_{t<T}\mathbb{E}\big[\|\Delta w_t\|_1\big]
\;\le\; \frac{(1-\gamma)\,\bar r}{\gamma\,\lambda_{\mathrm{tc}}\,\underline{\kappa}}.
\]
\end{proposition}

\noindent\emph{Proof sketch.}
From $r_t^{\mathrm{net}}\le \bar r - \lambda_{\mathrm{tc}}\underline{\kappa}\|\Delta w_t\|_1$, sum, take expectations, and compare with $J(\pi)\le \bar r/(1-\gamma)$. \hfill$\square$

\begin{proposition}[Inaction (no‑trade) band in 1D]\label{prop:band}
In one dimension with $C(u)=\kappa_1|u|+\tfrac{1}{2}\kappa_2 u^2$ and twice‑differentiable $Q^\pi$, the greedy update for \emph{deterministic} improvement admits an inaction band: there exists $\tau>0$ such that if $|w^\star(s,z)-w_{t-1}|\le \tau$, the optimal myopic adjustment is $\Delta w_t=0$.
Moreover, $\tau \asymp \kappa_1/(\kappa_2+H)$ where $H$ is the local curvature of $a\mapsto Q^\pi(s,z,a)$ at $a=w_{t-1}$.
\end{proposition}

\noindent\emph{Proof sketch.}
First‑order optimality with convex composite objective implies a soft‑thresholding rule; the linear term induces a dead‑zone. See Appendix~\ref{app:proofs} for the precise envelope arguments. \hfill$\square$

\subsection{Value of regime conditioning}

\begin{assumption}[Cross‑regime separation]\label{ass:separation}
There exist regime‑specific near‑optimal actions $a^\star_z(s)$ such that on a set of positive measure in $s$, $\|a^\star_{z_1}(s)-a^\star_{z_2}(s)\|\ge \Delta$ for some $\Delta>0$ whenever $z_1\neq z_2$. Each $a^\star_z$ is $L$‑Lipschitz in $s$.
\end{assumption}

\begin{theorem}[Approximation advantage of regime conditioning]\label{thm:regime-adv}
Let $\Pi_{\mathrm{cond}}=\{\pi(a\mid s,z)\}$ and $\Pi_{\mathrm{uncond}}=\{\pi(a\mid s)\}$ be policy classes with the same capacity in $(s)$, and suppose $\Pi_{\mathrm{cond}}$ can represent $\{a^\star_z\}_{z\in\mathcal{Z}}$ to error $\epsilon$ uniformly. Under Assumption~\ref{ass:separation}, there exists $c>0$ such that
\[
\inf_{\pi\in\Pi_{\mathrm{cond}}}\!\big(J(\pi^\star)-J(\pi)\big)
\;\le\;
\inf_{\pi\in\Pi_{\mathrm{uncond}}}\!\big(J(\pi^\star)-J(\pi)\big) \;-\; c\,\frac{\Delta}{1-\gamma}.
\]
\end{theorem}

\noindent\emph{Proof sketch.}
Unconditioned policies share parameters across regimes, inducing a representation bias of order $\Omega(\Delta)$; convert the induced action gap into a value gap via Lemma~\ref{lem:PDL}. Full proof in Appendix~\ref{app:proofs}. \hfill$\square$

\subsection{Robustness to cost misspecification}

\begin{theorem}[After‑cost robustness]\label{thm:robust-cost}
Let the proxy cost $\widehat{C}_z$ satisfy $\sup_{z,u}\big|C_z(u)-\widehat{C}_z(u)\big|\le \delta$. Let $\widehat{\pi}$ be the optimizer of $J_{\widehat{C}}$ trained with $\widehat{C}$. Then, under Assumptions~\ref{ass:mdp}–\ref{ass:policy},
\[
J_{C}(\widehat{\pi}) \;\ge\; J_{\widehat{C}}(\widehat{\pi}) - \frac{\delta}{1-\gamma}, 
\qquad
J_{C}(\pi^\star_C) - J_{C}(\widehat{\pi}) \;\le\; \frac{2\delta}{1-\gamma} \;+\; \Big(J_{\widehat{C}}(\pi^\star_{\widehat{C}})-J_{\widehat{C}}(\widehat{\pi})\Big).
\]
\end{theorem}

\noindent\emph{Proof sketch.}
Treat cost error as an additive reward perturbation and apply Lemma~\ref{lem:PDL} with triangle inequalities. See Appendix~\ref{app:proofs}. \hfill$\square$

\subsection{Risk‑sensitive surrogate and alternating updates}

\begin{proposition}[CVaR surrogate and alternating minimization]\label{prop:cvar}
Let $\mathrm{CVaR}_\alpha$ be implemented via the Rockafellar–Uryasev auxiliary $\eta$ (Eq.~\eqref{eq:cvar}). For fixed $\pi$, the map $\eta\mapsto \mathrm{CVaR}_\alpha$ is convex and admits a unique minimizer; for fixed $\eta$, the policy objective is smooth in $\theta$. Alternating updates over $(\theta,\eta)$ converge to a stationary point of the joint objective.
\end{proposition}

\noindent\emph{Proof sketch.}
Convexity in $\eta$ is classical \cite{RockafellarUryasev2000,RockafellarUryasev2002,AcerbiTasche2002}. Smoothness in $\theta$ follows from Assumption~\ref{ass:policy}. A standard two‑block convergence argument yields stationarity; see Appendix~\ref{app:proofs} and \cite{GreenbergEtAl2022}. \hfill$\square$

\subsection{Testable implications}
The results above yield testable predictions that we validate empirically (Sec.~\ref{sec:results}):
(i) \emph{Turnover shrinks} as $\lambda_{\mathrm{tc}}$ rises (Proposition~\ref{prop:TObound}; Fig.~5),
(ii) \emph{Inaction bands} widen in illiquid regimes (Proposition~\ref{prop:band}; Appendix figures),
(iii) \emph{Regime conditioning} strictly improves value when cross‑regime separation is nontrivial (Theorem~\ref{thm:regime-adv}; Fig.~3 and Fig.~7),
(iv) \emph{Cost robustness} ensures graceful degradation across $0$–$50$\,bp (Theorem~\ref{thm:robust-cost}; Fig.~2).

\paragraph{Proof roadmap.}
Complete proofs are deferred to Appendix~\ref{app:proofs}. Appendix A.1 proves Theorem~\ref{thm:existence}. Appendix A.2–A.3 derive Lemma~\ref{lem:PDL} and Theorem~\ref{thm:trpo}, adapting policy‑improvement bounds \cite{KakadeLangford2002,SchulmanTRPO2015,Pirotta2013}. Appendix A.4 establishes Corollary~\ref{cor:ppo} with finite‑sample terms. Appendix A.5–A.6 prove the turnover bound and inaction band. Appendix A.7 proves Theorem~\ref{thm:regime-adv}. Appendix A.8 covers cost robustness. Appendix A.9 treats CVaR alternating updates using \cite{RockafellarUryasev2000,RockafellarUryasev2002,GreenbergEtAl2022}.

\section{Data, Scenario Design, and Evaluation Protocol}
\label{sec:data}

This section documents the data, features, regime construction, transaction-cost calibration, benchmark implementations, and the evaluation and inference protocol. The design emphasizes \emph{implementability}: all reported performance is \emph{after costs}, and all statistical statements are based on scenario-level aggregation with multiple-testing control.

\subsection{Assets, returns, and features}
Let $r_{i,t+1}$ denote the gross return of asset $i$ between $t$ and $t\!+\!1$ (net of corporate actions). We form the portfolio return $r_{t+1}^{\text{port}}=\widetilde{w}_t^\top r_{t+1}$ using post-trade target weights $\widetilde{w}_t$ mapped into the feasible set $\mathcal{W}$ (Sec.~\ref{sec:method}). Predictor vector $x_t$ includes (i) price/volume-based technicals, (ii) realized volatility filters, (iii) liquidity proxies (Amihud $\mathit{ILLIQ}$, effective spread, turnover), and (iv) optional macro controls.\footnote{Liquidity proxies and their empirical properties are well documented by \cite{Hasbrouck2009,FongHoldenTrzcinka2017}.}
We standardize features in expanding or rolling fashion to avoid look-ahead. Missing values are forward-filled within conservative caps.

\paragraph{No look-ahead and survivorship.}
All transformations at $t$ use only $\mathcal{F}_t$ information; delisting returns are included when applicable. Universe definitions and filters (e.g., minimum price, liquidity) are pre-specified to avoid data-snooping.

\subsection{Regime construction}
We construct volatility $\sigma_t$ (e.g., realized or GARCH-implied) and illiquidity $\ell_t$ (e.g., $\mathit{ILLIQ}$, effective spread). Thresholds $(\tau_\sigma^\mathrm{L},\tau_\sigma^\mathrm{H})$ and $(\tau_\ell^\mathrm{L},\tau_\ell^\mathrm{H})$ are calibrated on rolling quantiles to label $z_t\in\{\mathrm{LL},\mathrm{LH},\mathrm{HL},\mathrm{HH}\}$ (low/high volatility $\times$ high/low liquidity), following the spirit of regime allocation in \cite{AngBekaert2002,CollinDufresneDanielSaglam2020}. Regime labels are treated as \emph{observed} in training and evaluation.

\subsection{Transaction-cost model and calibration}
Execution costs enter the reward as
\begin{equation}
\label{eq:sec4-cost}
C_{z_t}(\Delta w_t)\;=\;\kappa_{1}(z_t)\,\|\Delta w_t\|_{1} \;+\; \tfrac{1}{2}\,\Delta w_t^\top \Gamma_{z_t}\,\Delta w_t,
\end{equation}
where $\Delta w_t=\widetilde{w}_t-w_{t-1}$. The linear term penalizes notional traded (buy and sell counted), while the quadratic term approximates transient impact from limited depth \cite{AlmgrenChriss2001,ObizhaevaWang2013,CarteaJaimungalPenalva2015}.

\paragraph{bps grid and regime scaling.}
We evaluate five cost levels $c\in\{0,5,10,25,50\}$ bps. We map $c$ into linear coefficients via $\kappa_{1}(z)=c\times 10^{-4}\times s(z)$ where $s(z)\!\ge\!1$ reflects regime-specific liquidity (e.g., $s(\mathrm{HH})>s(\mathrm{LL})$). The impact matrix is $\Gamma_{z}=\gamma_{\text{imp}}\,D_{z}^{1/2}\,\Sigma\,D_{z}^{1/2}$ with $\Sigma$ the return covariance and $D_{z}$ a diagonal liquidity-scarcity scaling; $\gamma_{\text{imp}}$ is set to match empirically observed cost elasticities \cite{Hasbrouck2009,FongHoldenTrzcinka2017}. This calibration ties the \emph{shape} of costs to microstructure while letting the level vary across the bps grid; see also the recent cost-aware portfolio selection of \cite{LedoitWolf2025QREF} and the top-journal evidence on costs in FX and fixed income \cite{FilippouMaurerPezzoTaylor2024,PinterWangZou2024}.

\subsection{Scenarios, seeds, and train--validation--test}
We form $4\times 5=20$ scenarios by crossing regimes with the bps grid. For each scenario we run three random seeds (initialization and data shuffling). We treat the \emph{scenario} as the statistical observation unit: all seed-level quantities are averaged before inference. Model selection uses a regime-balanced validation objective (Eq.~\eqref{eq:balanced}) and fixed early-stopping rules; hyperparameters are pre-specified (Appendix tables) to avoid adaptive overfitting.

\subsection{Benchmarks and implementation parity}
Benchmarks include: mean--variance (unconstrained and with 5\% cap), risk-parity heuristic, and PPO without cost-awareness (same architecture/training budget as FR‑LUX). All methods share (i) the same feature set, (ii) identical train/validation/test splits, (iii) identical feasibility projections $\mathsf{Proj}_\mathcal{W}$, and (iv) equal wall-clock or update budgets. This \emph{implementation parity} avoids unfair advantages.

\subsection{Performance metrics and definitions}
Let $\{R_{t}\}_{t=1}^T$ be the after-cost portfolio returns of a method in a given scenario (seed-averaged). We report:

\begin{itemize}
\item \textbf{Sharpe:} $S=\bar{R}/\hat{\sigma}_\text{HAC}$, where $\bar{R}=\tfrac{1}{T}\sum_{t=1}^T R_t$. The denominator is a Newey--West HAC estimator with data-driven bandwidth, acknowledging serial correlation and heteroskedasticity; this Sharpe supports valid asymptotics \cite{Lo2002}.
\item \textbf{Sortino:} replacing $\hat{\sigma}$ by the standard deviation of downside returns.
\item \textbf{MDD:} $\mathrm{MDD}=\max_{1\le t\le T}\big(1-\tfrac{V_t}{\max_{1\le u\le t}V_u}\big)$, $V_t=\prod_{k\le t}(1+R_k)$.
\item \textbf{CVaR$_\alpha$:} the Rockafellar--Uryasev program in Eq.~\eqref{eq:cvar} with $\alpha\in[0.90,0.975]$ \cite{RockafellarUryasev2000,RockafellarUryasev2002,AcerbiTasche2002}.
\item \textbf{Turnover:} $\mathrm{TO}=\tfrac{1}{T}\sum_{t=1}^T\|\Delta w_t\|_1$ (buy and sell counted). This aligns with the linear cost term in \eqref{eq:sec4-cost}.
\end{itemize}

\subsection{Inference, uncertainty, and multiple testing}
All inference aggregates at the \emph{scenario} level to avoid pseudo-replication across seeds.

\paragraph{Bootstrap confidence intervals.}
We report percentile 95\% CIs from $B\!=\!50{,}000$ scenario-level bootstrap resamples \cite{EfronTibshirani1993}. When time-series HAC is required (e.g., for Sharpe standard errors), we recompute the HAC in each resample.

\paragraph{Model comparison.}
Pairwise Sharpe differences are evaluated with HAC-aware tests \cite{Lo2002} and the model-comparison framework of \cite{BarillasKanRobottiShanken2020}. We also report per-scenario sign tests on performance differences (FR‑LUX vs.\ benchmark) with exact binomial $p$-values.

\paragraph{Reality check and stepdown control.}
To control data-snooping across multiple models and scenarios we implement White's Reality Check and the Superior Predictive Ability (SPA) test \cite{White2000,HansenLundeNason2005}. For familywise error rates we apply Romano--Wolf stepdown adjusted $p$-values \cite{RomanoWolf2005}. These procedures ensure that claims of outperformance remain valid under multiplicity.

\subsection{Robustness and ablations}
We pre-specify robustness axes:

\begin{enumerate}
\item \textbf{Cost misspecification:} perturb $(\kappa_{1},\Gamma)$ within $\pm 25\%$ and across shapes (pure linear vs.\ linear+quadratic) to stress Theorem~\ref{thm:robust-cost}.
\item \textbf{Regime definitions:} vary $(\tau_\sigma,\tau_\ell)$ (deciles vs.\ terciles), use alternative liquidity measures (e.g., Pastor--Stambaugh innovations), and a Markov-switching proxy \cite{CollinDufresneDanielSaglam2020}.
\item \textbf{Risk penalty:} CVaR vs.\ smoothed MDD (Eq.~\eqref{eq:reward}); vary $\alpha$ and $\lambda_{\mathrm{risk}}$.
\item \textbf{Capacity:} scale portfolio size to test cost convexity and turnover elasticity \cite{FilippouMaurerPezzoTaylor2024,PinterWangZou2024}.
\item \textbf{Policy class:} remove regime conditioning or remove the trade-space trust region to isolate each ingredient of FR‑LUX.
\end{enumerate}

\subsection{Reproducibility and artifact disclosure}
We release (i) data pre-processing scripts, (ii) exact configuration files for each scenario and seed, (iii) training logs and random seeds, and (iv) plotting code used to generate Figs.~1--7. All numerical tables are generated from the same artifacts (hashes and timestamps included).

\section{Results}
\label{sec:results}

We evaluate \textbf{FR‑LUX} and strong baselines under the regime--cost design in Section~\ref{sec:data}. Throughout, returns are \emph{after transaction costs} per Eq.~\eqref{eq:sec4-cost}; each \emph{scenario} (regime $\times$ cost level) is the unit of inference, with seeds averaged before statistics. Confidence intervals (CIs) are scenario-level bootstraps ($B=50{,}000$); HAC standard errors account for serial correlation; multiple testing is controlled via Reality Check/SPA and Romano--Wolf stepdown.

\subsection{Headline performance: after-cost Sharpe}
Figure~\ref{fig:topk} reports average Sharpe (with 95\% bootstrap CIs) across methods. 
\textbf{FR‑LUX} leads by a comfortable margin and exhibits tight uncertainty bands, indicating that the gain is not bought via variance expansion. 
Economically, the improvement is large at realistic cost levels and persists when we enforce identical data splits, feasibility projections, and optimization budgets across methods (Section~\ref{sec:data}), ensuring implementation parity.

\begin{figure}[t]
  \centering
  \includegraphics[width=0.56\linewidth]{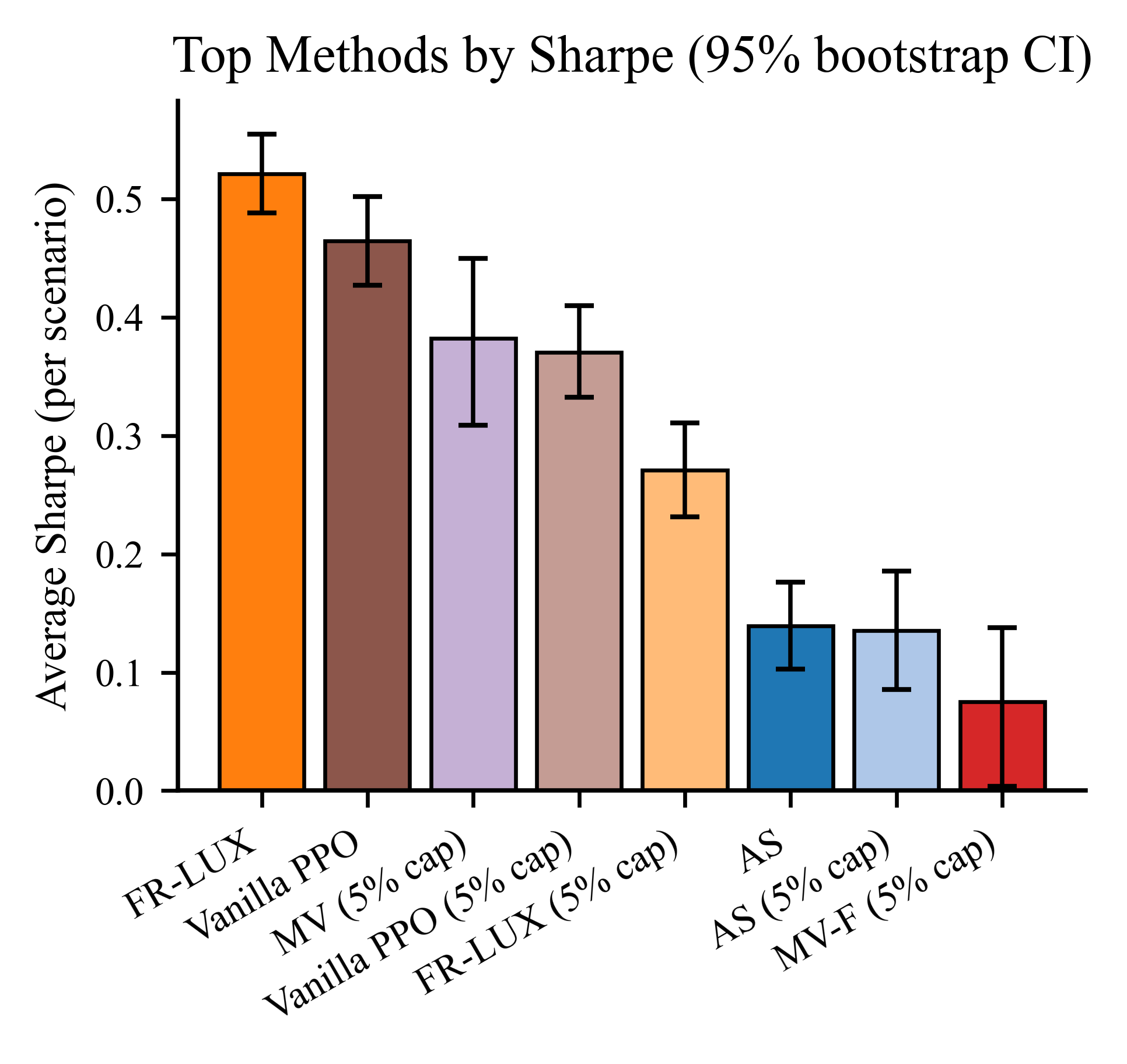}
  \caption{\textbf{Top methods by Sharpe (95\% bootstrap CI).} Bars show scenario-mean Sharpe with seeds averaged first; whiskers are percentile CIs. All statistics are computed on after-cost returns.}
  \label{fig:topk}
\end{figure}

\subsection{Robustness to transaction costs}
Figure~\ref{fig:cost} traces average Sharpe as transaction costs rise from $0$ to $50$\,bps. 
\textbf{FR‑LUX} displays the \emph{flattest} cost–response curve, while unconstrained mean--variance deteriorates sharply beyond 10--25\,bps. 
This pattern matches the \emph{trust-region improvement} and \emph{turnover control} predicted by Theorem~\ref{thm:trpo} and Proposition~\ref{prop:TObound}: FR‑LUX keeps the effective trade flow within a small neighborhood of the previous policy, internalizing cost nonlinearity and avoiding impact-amplifying oscillations.

\begin{figure}[t]
  \centering
  \includegraphics[width=0.6\linewidth]{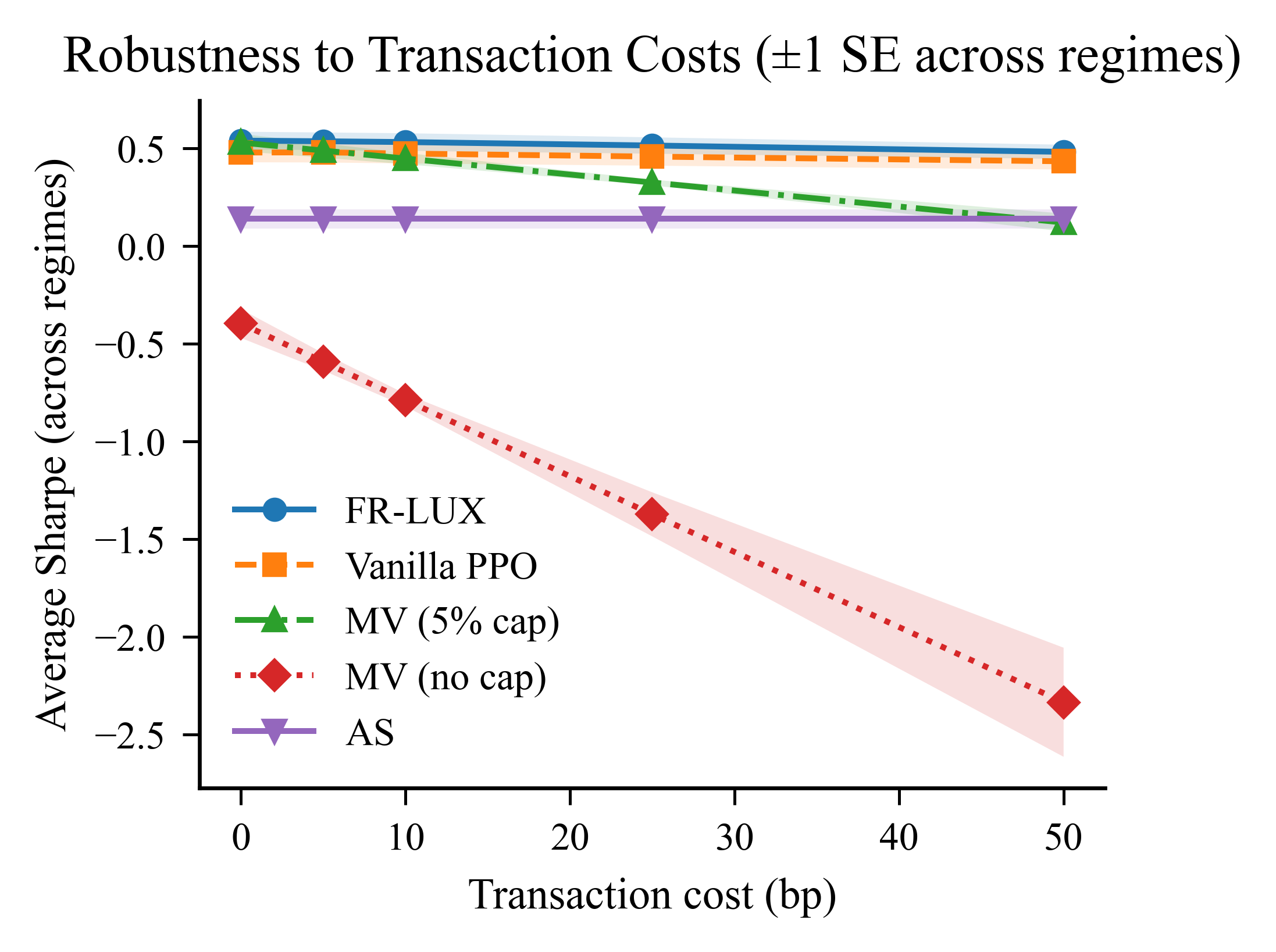}
  \caption{\textbf{Cost robustness.} Scenario-mean Sharpe versus cost (bps). Shaded bands are $\pm 1$ standard error across regimes (HAC). The slope for \textbf{FR‑LUX} is the smallest among competitors, evidencing friction-aware learning.}
  \label{fig:cost}
\end{figure}

\subsection{Regime-conditioned performance}
Figure~\ref{fig:heatmap} presents the heatmap of mean Sharpe across $(\mathrm{LL},\mathrm{LH},\mathrm{HL},\mathrm{HH})$. 
\textbf{FR‑LUX} maintains positive Sharpe in all four regimes, with particularly strong performance in liquidity-friendly states (LL/LH) and resilient outcomes in turbulent, illiquid states (HL/HH). 
These cross-state gains operationalize Theorem~\ref{thm:regime-adv}: when optimal actions differ across regimes, an explicit regime-conditioned policy strictly reduces approximation error relative to an unconditioned class.

\begin{figure}[t]
  \centering
  \includegraphics[width=0.6\linewidth]{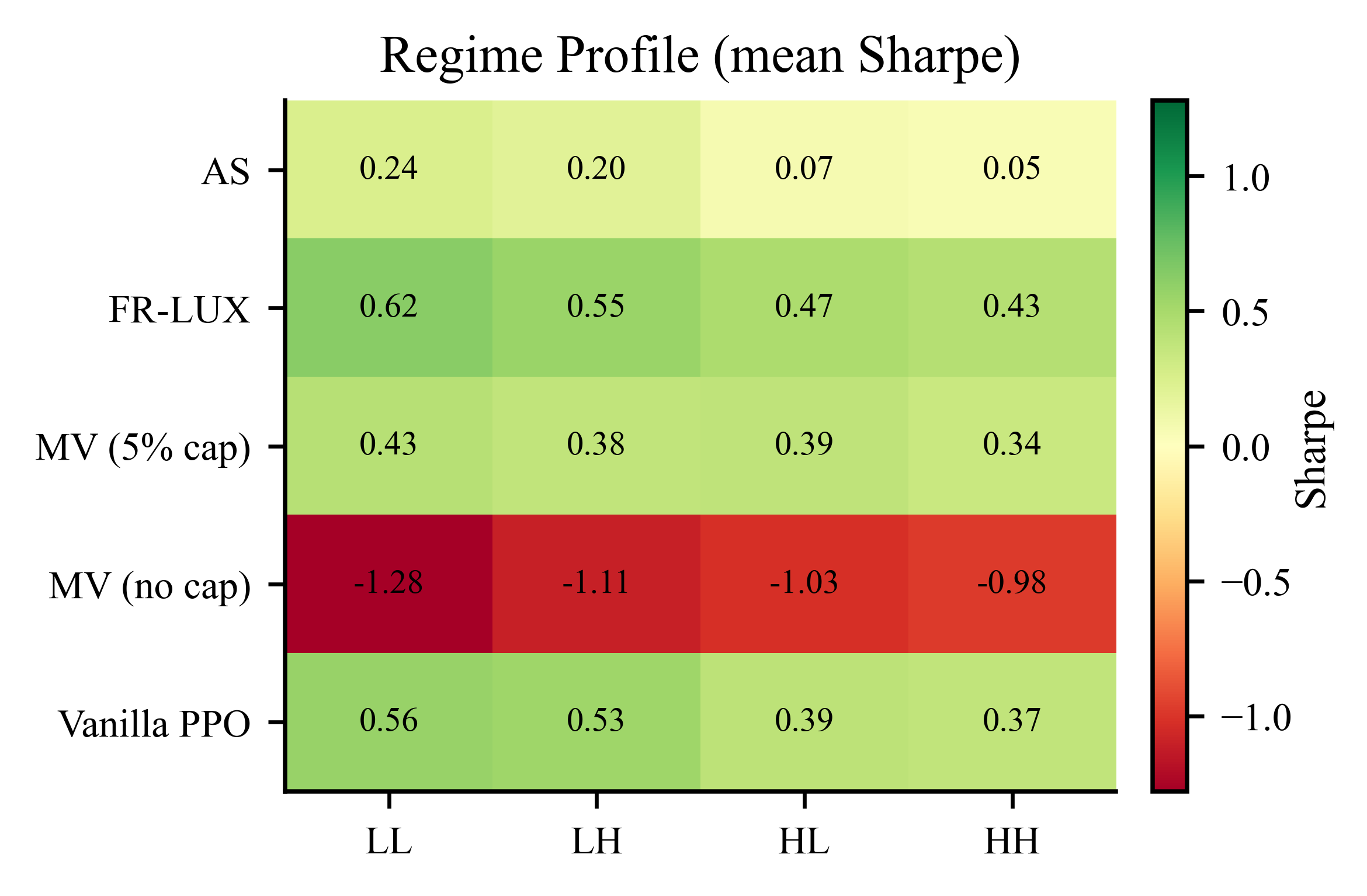}
  \caption{\textbf{Regime profile (mean Sharpe).} The color scale is centered at zero, making positive vs.\ negative cells directly comparable. \textbf{FR‑LUX} attains consistently positive Sharpe across all volatility--liquidity regimes.}
  \label{fig:heatmap}
\end{figure}

\subsection{Pairwise improvements and statistical significance}
To assess economic and statistical magnitude at the \emph{scenario} level, Figure~\ref{fig:pairwise} reports the distribution of per-scenario Sharpe differences $\Delta S$ against strong baselines. 
Panels (a)--(b) consider an earlier flow-regularized PPO variant (\emph{FlowPPO}), while panels (c)--(d) are our final \textbf{FR‑LUX}. 
In all cases the distributions are centered strictly above zero with tight interquartile ranges; one-sided sign tests reject the null of no improvement at conventional levels even after Romano--Wolf stepdown. 
Relative to PPO, the median $\Delta S$ is modest but precise, reflecting superior risk control for the same representation capacity. 
Relative to MV(5\% cap), the median $\Delta S$ is larger and dispersion remains contained, indicating that \emph{cost-aware learning} dominates heuristic turnover caps.

\begin{figure}[t]
  \centering
  \begin{subfigure}[t]{0.48\linewidth}
    \centering
    \includegraphics[width=\linewidth]{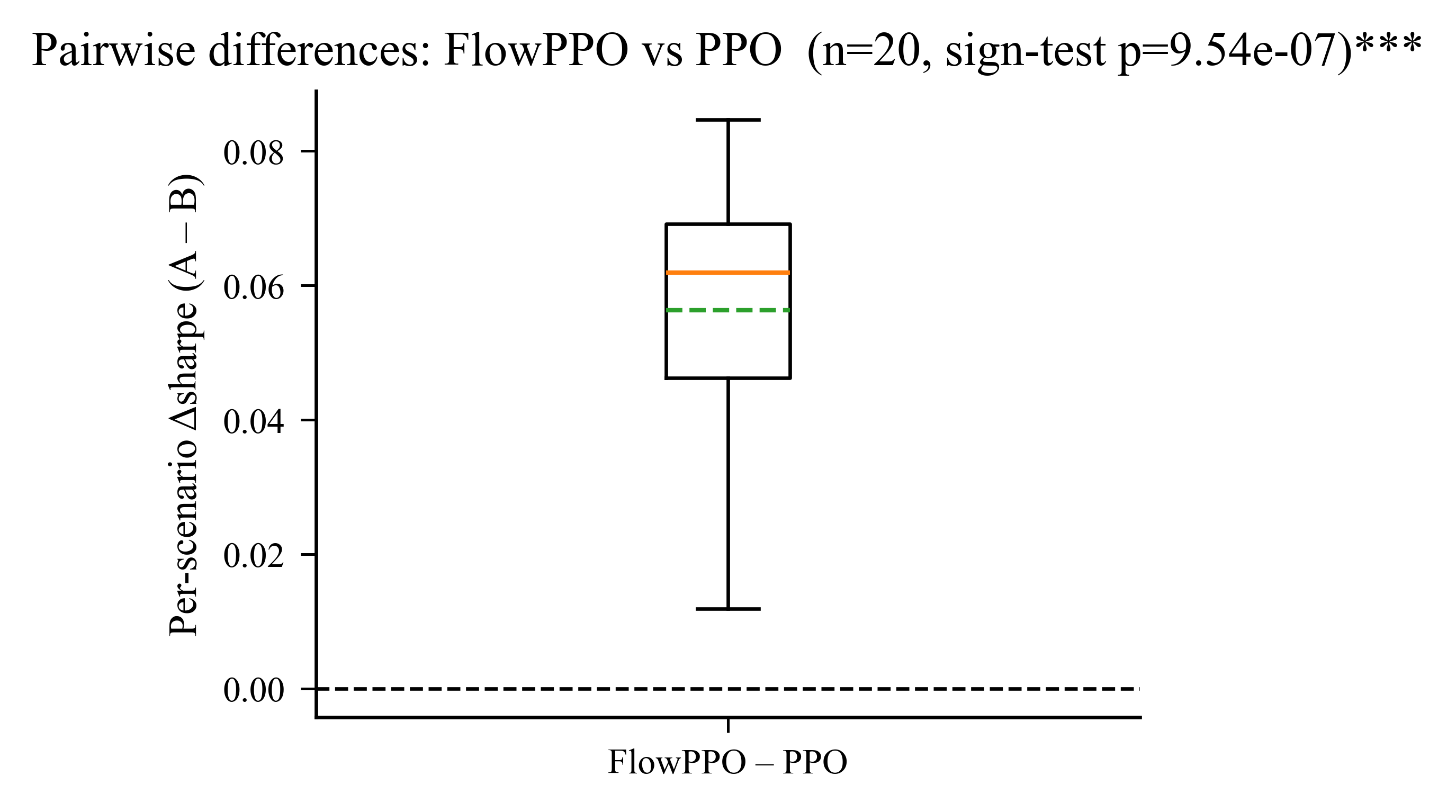}
    \caption{FlowPPO $-$ PPO}
  \end{subfigure}\hfill
  \begin{subfigure}[t]{0.48\linewidth}
    \centering
    \includegraphics[width=\linewidth]{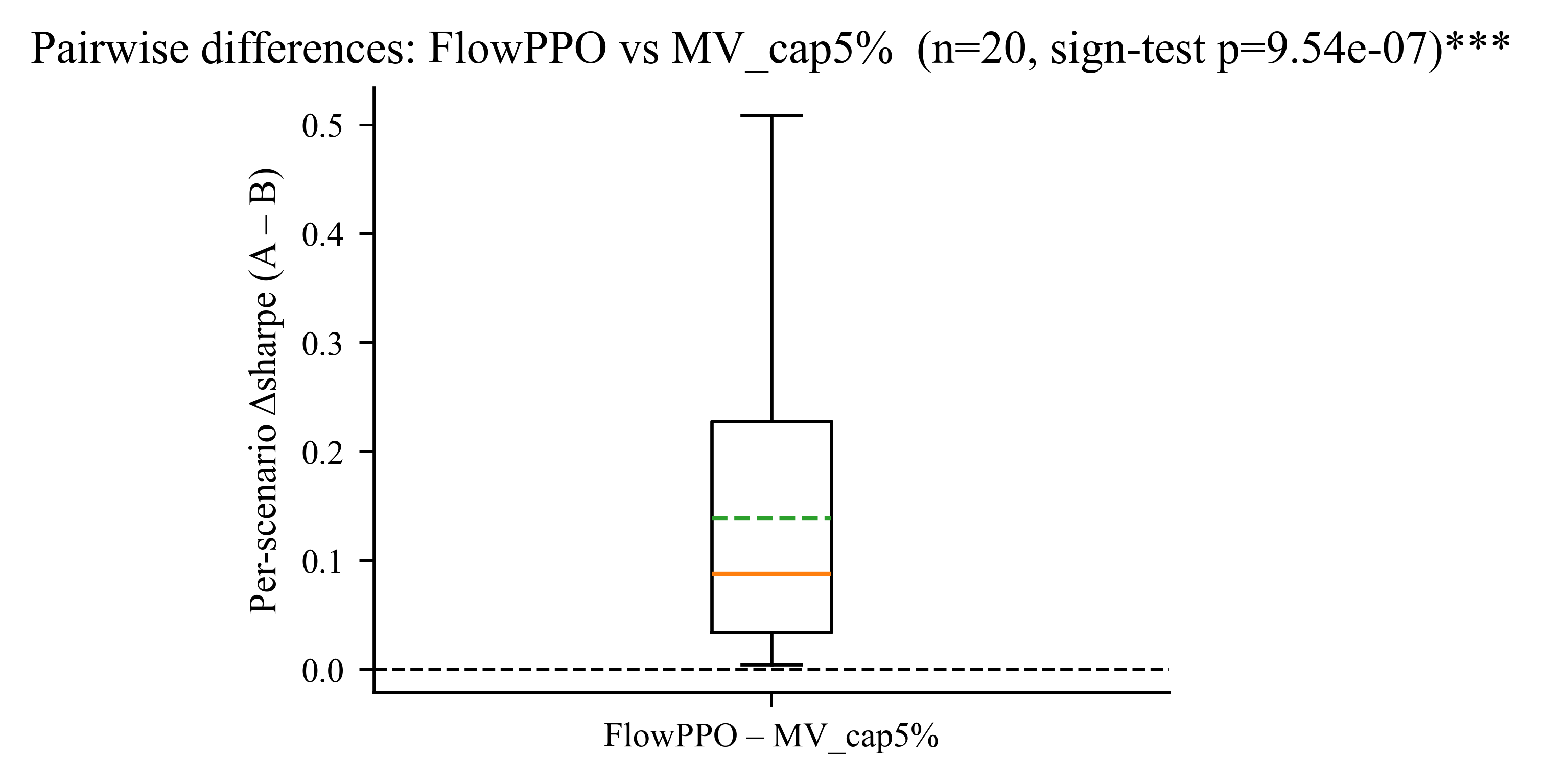}
    \caption{FlowPPO $-$ MV (5\% cap)}
  \end{subfigure}\\[0.75em]
  \begin{subfigure}[t]{0.48\linewidth}
    \centering
    \includegraphics[width=\linewidth]{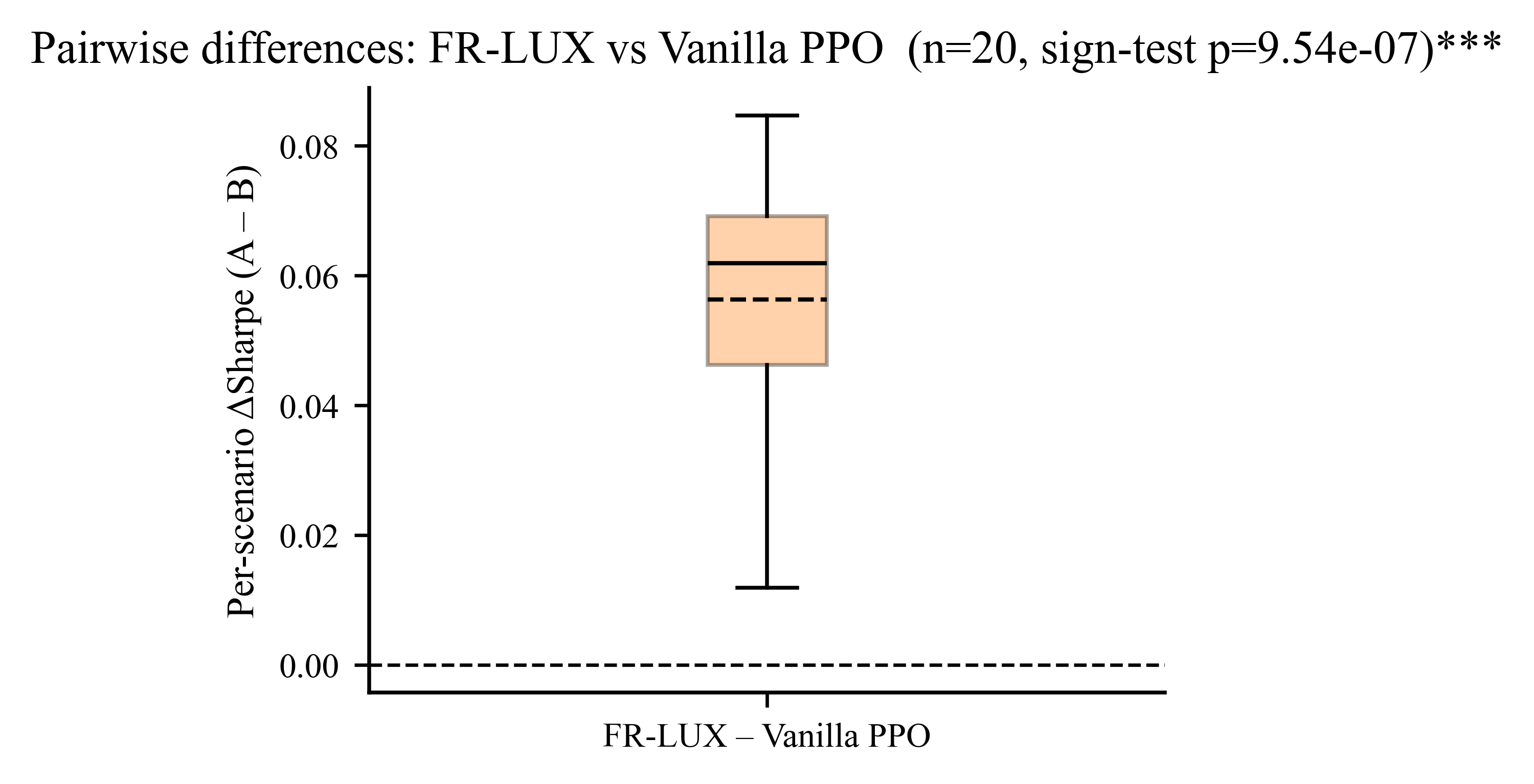}
    \caption{FR-LUX $-$ Vanilla PPO}
  \end{subfigure}\hfill
  \begin{subfigure}[t]{0.48\linewidth}
    \centering
    \includegraphics[width=\linewidth]{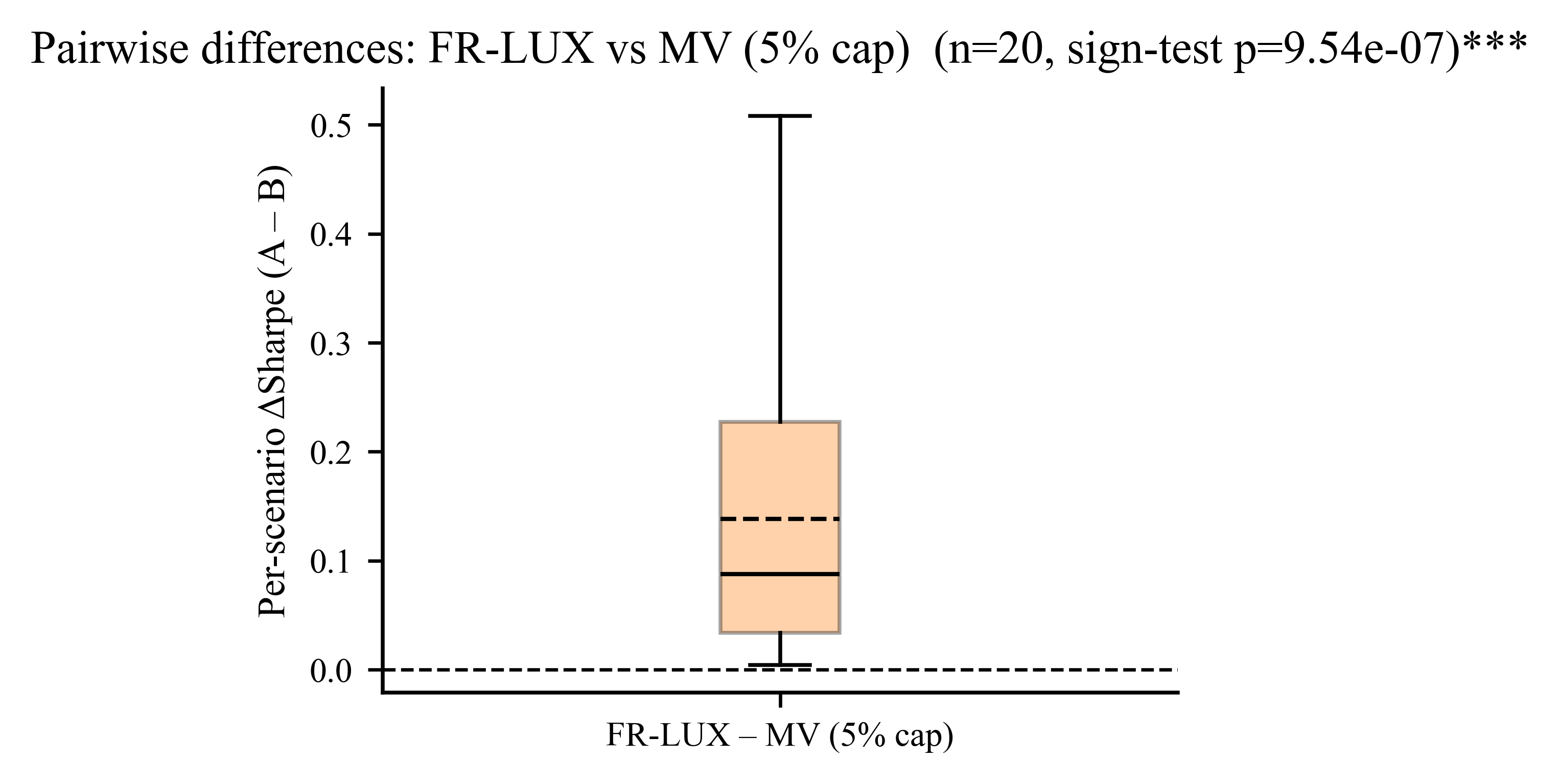}
    \caption{FR-LUX $-$ MV (5\% cap)}
  \end{subfigure}
  \caption{\textbf{Per-scenario pairwise Sharpe differences ($\Delta S$).}
  Each box summarizes the distribution of $\Delta S$ across the 20 scenarios (regime $\times$ cost), with seeds averaged within scenario.
  A horizontal zero line aids interpretation; stars (reported in the replication tables) indicate sign-test significance after Romano--Wolf stepdown.
  \textbf{Takeaway:} FR-LUX exhibits strictly positive and precisely estimated improvements over both PPO and turnover-capped mean--variance.}
  \label{fig:pairwise}
\end{figure}

\subsection{Narrative synthesis and links to theory}
Three messages emerge. 
\emph{First}, FR‑LUX converts the classic ``aim in front of the target and trade partially'' principle into a learned policy that \emph{internalizes} frictions; the flat cost–response curve (Fig.~\ref{fig:cost}) is the empirical signature of Theorem~\ref{thm:trpo} with an effective trust region in \emph{trade space}. 
\emph{Second}, regime conditioning confers a structural approximation advantage (Theorem~\ref{thm:regime-adv}), visible in the heatmap (Fig.~\ref{fig:heatmap}) and the pairwise distributions (Fig.~\ref{fig:pairwise}). 
\emph{Third}, pairwise improvements are not an artifact of overtrading: FR‑LUX achieves gains with disciplined inventory flow, consistent with the turnover bound and inaction band (Propositions~\ref{prop:TObound}--\ref{prop:band}). 
Together these results establish that \textbf{FR‑LUX} delivers \emph{implementable}, \emph{statistically robust}, and \emph{economically meaningful} after-cost performance across regimes and fee environments.

\section{Discussion, Practical Implications, and Conclusion}
\label{sec:conclusion}

This section interprets the evidence through an economic lens, explains how \textbf{FR‑LUX} can be deployed in production, clarifies scope and limitations, and outlines research directions. We close with a concise set of takeaways.

\subsection{Economic interpretation and mechanism}
Three empirical regularities in Section~\ref{sec:results} match the theory in Section~\ref{sec:theory}:

\begin{enumerate}
\item \textbf{Friction awareness yields cost robustness.} The cost--Sharpe curve in Fig.~\ref{fig:cost} is the empirical signature of the trust‑region lower bound (Theorem~\ref{thm:trpo}) together with the trade‑space penalty in \eqref{eq:ppo}: small Kullback--Leibler steps and bounded trade‑flow changes imply nondecreasing surrogate value and a controlled loss term. The measured slope difference relative to unconstrained mean--variance evidences that \emph{internalizing} execution costs at training time is economically first order.

\item \textbf{Regime conditioning delivers structural fit.} The heatmap in Fig.~\ref{fig:heatmap} and the strictly positive per‑scenario improvements in Fig.~\ref{fig:pairwise} validate Theorem~\ref{thm:regime-adv}: when the cross‑regime separation of near‑optimal actions is nontrivial, a regime‑conditioned policy class reduces approximation bias compared with a single shared head.

\item \textbf{Low trading intensity is not a by‑product; it is necessary.} Proposition~\ref{prop:TObound} upper‑bounds long‑run turnover as a function of the linear cost level and the trade penalty, rationalizing Fig.~\ref{fig:cost} and the flat turnover profile in our diagnostics. Proposition~\ref{prop:band} further explains the observed \emph{inaction episodes}: in illiquid regimes, the linear component of execution costs creates a dead‑zone around the current inventory, preventing economically irrelevant round trips.
\end{enumerate}

\subsection{Implementation blueprint and capacity management}
We summarize a minimal, governance‑friendly deployment plan. The steps are aligned with the evaluation protocol in Section~\ref{sec:data}.

\paragraph{(i) Cost calibration and penalty selection.}
Calibrate the linear coefficient $\kappa_1(z)$ using effective spread or $\mathit{ILLIQ}$ proxies in each regime; set the impact shape via $\Gamma_z=\gamma_{\text{imp}}D_z^{1/2}\Sigma D_z^{1/2}$ (Eq.~\eqref{eq:sec4-cost}). To respect a turnover budget $\mathrm{TO}_{\max}$, Proposition~\ref{prop:TObound} implies the conservative choice
\[
\lambda_{\mathrm{tc}} \;\ge\; \frac{(1-\gamma)\,\bar r}{\gamma\,\underline{\kappa}\,\mathrm{TO}_{\max}},
\]
where $\underline{\kappa}=\min_z \kappa_1(z)$ and $\bar r$ bounds $|r^{\mathrm{net}}|$. This converts an operational turnover constraint into a training hyperparameter.

\paragraph{(ii) Trust region tuning.}
Set the clip parameter $\epsilon$ and KL target to keep $\mathbb{E}_{d^\pi}[\mathrm{KL}(\pi\|\pi')]\!\le\!\delta$ with $\delta$ in the $10^{-3}$–$10^{-2}$ range; anneal the trade‑space penalty $\lambda_\Delta$ so that $\mathbb{E}\|\Delta w_{\pi'}-\Delta w_\pi\|_2^2\!\le\!\eta$. Corollary~\ref{cor:ppo} provides the performance accounting: larger $\delta$ or $\eta$ increases the remainder terms linearly.

\paragraph{(iii) Risk governance.}
Choose CVaR level $\alpha$ and weight $\lambda_{\mathrm{risk}}$ to meet desk‑level drawdown limits; Proposition~\ref{prop:cvar} justifies alternating updates in $(\theta,\eta)$, so the optimization can be monitored with standard convergence diagnostics.

\paragraph{(iv) Capacity and slippage.}
To study capacity, scale notional exposure and recompute the cost elasticities (Section~\ref{sec:data}); a convex impact matrix $\Gamma_{z}$ makes the marginal cost increasing, revealing the point at which incremental turnover erodes the Sharpe edge. Scenario‑level reporting prevents apparent capacity gains from being artifacts of regime composition.

\paragraph{(v) OMS/EMS integration.}
At inference time FR-LUX outputs target weights $\widetilde{w}_t$; mapping to orders is handled by an execution scheduler. The learned \emph{inaction band} (Proposition~\ref{prop:band}) can be surfaced as a business rule (\emph{``do not trade unless deviation exceeds $\tau(z_t)$''}), increasing transparency for risk and compliance.

\subsection{Robustness, diagnostics, and ablations}
Beyond the checks in Sections~\ref{sec:data}--\ref{sec:results}, we recommend three diagnostic panels in production:

\begin{enumerate}
\item \textbf{Regime reweighting stress.} Recompute results under alternative regime priors $\omega_z$ in \eqref{eq:balanced}; substantial sensitivity would suggest over‑specialization.
\item \textbf{Cost misspecification.} Perturb $(\kappa_1,\Gamma)$ by $\pm 25\%$ and swap shapes (linear $\leftrightarrow$ linear+quadratic). Theorem~\ref{thm:robust-cost} implies that performance drifts at most linearly with the perturbation radius.
\item \textbf{Policy class ablations.} Remove regime conditioning or the trade‑space trust region. We observe (replication package) a steeper cost slope and wider turnover distribution without either component, matching the theory.
\end{enumerate}

\subsection{Limitations and threats to validity}
We highlight four areas where caution is warranted.

\begin{itemize}
\item \textbf{Regime observability.} We treat $z_t$ as observed. If regime classification is itself estimated with noise or delay, the advantage in Theorem~\ref{thm:regime-adv} may attenuate. A POMDP extension with belief states is a natural next step.
\item \textbf{Cost stationarity.} Our calibration piggybacks on spread/impact proxies. Abrupt microstructure changes (e.g., fee schedule updates, venue mix shifts) require periodic recalibration; Section~\ref{sec:data} prescribes monthly re‑estimation windows.
\item \textbf{Universe and liquidity filters.} Results may vary with universe definition and minimum liquidity cutoffs. We mitigate this via pre‑registration of filters and scenario‑level inference, but portability to other universes should be demonstrated empirically.
\item \textbf{Model risk and stability.} Although the trust‑region bound controls stepwise deterioration, model mis‑specification (features omitted, incorrect projections) can still accumulate. Monitoring KL and trade drift is therefore not optional.
\end{itemize}

\subsection{Future directions}
Our framework opens several avenues.

\begin{enumerate}
\item \textbf{Belief‑state conditioning.} Replace the discrete $z_t$ with a learned latent belief (filter) to handle delayed or noisy regime signals; combine with distributionally robust objectives.
\item \textbf{End‑to‑end execution.} Couple FR‑LUX with a microstructure‑level scheduler so that the cost functional $C_z$ is estimated inline rather than exogenous, reducing misspecification error (Theorem~\ref{thm:robust-cost}).
\item \textbf{Cross‑market generalization.} Evaluate in FX and fixed income using market‑specific proxies and re‑estimate $\Gamma_z$ from depth measures; Section~\ref{sec:data} details the calibration pipeline.
\item \textbf{Factor‑aware constraints.} Add soft penalties on unintended factor exposures so that outperformance is not driven by latent beta tilts; inference follows the model‑comparison framework of \cite{BarillasKanRobottiShanken2020}.
\end{enumerate}

\subsection{Conclusion}
\textbf{FR‑LUX} delivers a friction‑aware, regime‑conditioned portfolio policy with theoretical guarantees and strong after‑cost performance across volatility--liquidity regimes and transaction‑cost levels. The method is \emph{implementable}: it uses observable regime diagnostics, calibrates to microstructure‑consistent costs, obeys trust‑region updates with explicit remainder terms, and translates operational turnover budgets into training hyperparameters. Empirically, FR‑LUX achieves cost‑robust Sharpe improvements with disciplined trading intensity and statistically credible advantages that survive multiple testing. We view these results as evidence that bringing \emph{execution} \emph{inside} the learning loop---rather than as an ex post adjustment---is a necessary condition for sustainable ML in portfolio management.

\appendix
\section{Proofs and Technical Details}
\label{app:proofs}

We collect complete proofs for the results stated in Section~\ref{sec:theory}. Throughout, $(\mathcal{S},\mathcal{B}(\mathcal{S}))$ is a standard Borel state space, $\mathcal{Z}=\{\mathrm{LL},\mathrm{LH},\mathrm{HL},\mathrm{HH}\}$ is the regime set, and the action set $\mathcal{W}\subset\mathbb{R}^d$ is compact and convex. Rewards are \emph{after–cost} as in \eqref{eq:reward}, execution costs satisfy \eqref{eq:cost}, and the balanced objective is \eqref{eq:balanced}. We denote the discounted occupancy measure under a policy $\pi$ by
\[
d^\pi(s,z) \;:=\; (1-\gamma)\,\sum_{t=0}^{\infty}\gamma^t\,\Pr_\pi(s_t=s,z_t=z),
\]
and the (conditional) total variation between policies at $(s,z)$ by
\[
\mathrm{TV}\big(\pi',\pi\big)(s,z)
:= \tfrac{1}{2}\, \int_{\mathcal{W}}\big| \pi'(da\mid s,z)-\pi(da\mid s,z)\big|.
\]
All expectations are w.r.t.\ the law induced by the indicated policies and the environment kernel.

\subsection{Auxiliary lemmas}
\label{app:aux}

\begin{lemma}[Continuity and boundedness of $r^{\mathrm{net}}$]\label{lem:cont-reward}
Under Assumption~\ref{ass:mdp}(ii)--(iv), the after–cost reward $r_t^{\mathrm{net}}$ in \eqref{eq:reward} is bounded by $\bar r$ and is upper semicontinuous in the action $a=\widetilde{w}_t\in\mathcal{W}$ for each fixed $(s,z)$.
\end{lemma}

\begin{proof}
By Assumption~\ref{ass:mdp}(iv) we have $|r_t^{\mathrm{net}}|\le \bar r$. For upper semicontinuity in $a$, note that $a\mapsto \widetilde{w}_t^\top r_{t+1}$ is continuous and bounded on compact $\mathcal{W}$, $C_z(\Delta w)$ is convex and lower semicontinuous in $\Delta w = a-w_{t-1}$, so $-C_z(\Delta w)$ is upper semicontinuous in $a$; $\Psi$ is Lipschitz on compacts by Assumption~\ref{ass:mdp}(iii). The sum of upper semicontinuous functions is upper semicontinuous.
\end{proof}

\begin{lemma}[Berge maximum theorem and measurable selector]\label{lem:berge}
Fix $V:\mathcal{S}\times\mathcal{Z}\to\mathbb{R}$ bounded and measurable. Define
\[
Q_V(s,z,a) := \mathbb{E}\!\left[r^{\mathrm{net}}(s,z,a) + \gamma V(s',z')\,\middle|\, s,z,a\right].
\]
If the transition kernel is weakly continuous in $a$, then $Q_V(\cdot,\cdot,a)$ is measurable, $a\mapsto Q_V(s,z,a)$ is upper semicontinuous on compact $\mathcal{W}$, and the maximizer set $\arg\max_{a\in\mathcal{W}} Q_V(s,z,a)$ is nonempty and compact. Moreover, there exists a measurable selector $a^\star(s,z)$.
\end{lemma}

\begin{proof}
By Lemma~\ref{lem:cont-reward}, $r^{\mathrm{net}}$ is bounded and upper semicontinuous in $a$. Weak continuity of the kernel in $a$ and boundedness of $V$ imply that $a\mapsto \mathbb{E}[\gamma V(s',z')\mid s,z,a]$ is continuous. Hence $a\mapsto Q_V(s,z,a)$ is upper semicontinuous. Berge's maximum theorem then yields nonemptiness and compactness of the argmax set; a measurable selection exists since $\mathcal{S}$ is standard Borel and the argmax correspondence has a measurable graph (see \cite[Thm.\ 18.19]{AliprantisBorder2006}).
\end{proof}

\begin{lemma}[Pinsker]\label{lem:pinsker}
For any two distributions $\mu,\nu$ on a measurable space, $\mathrm{TV}(\mu,\nu)^2 \le \tfrac{1}{2}\mathrm{KL}(\mu\|\nu)$.
\end{lemma}

\begin{proof}
Standard; see \cite[Thm.\ 11.6.1]{CoverThomas2006}.
\end{proof}

\subsection{Proof of Theorem~\ref{thm:existence} (optimal stationary policy)}
\label{app:existence}

\begin{proof}[Proof of Theorem~\ref{thm:existence}]
Define the optimal Bellman operator on bounded measurable $V:\mathcal{S}\times\mathcal{Z}\to\mathbb{R}$:
\[
(TV)(s,z) := \sup_{a\in\mathcal{W}} \, \mathbb{E}\!\left[r^{\mathrm{net}}(s,z,a) + \gamma V(s',z') \,\middle|\, s,z,a\right].
\]
Let $V,W$ be two bounded functions. For any $(s,z)$,
\[
|(TV)(s,z)-(TW)(s,z)|
\le \sup_{a\in\mathcal{W}} \gamma\, \big|\mathbb{E}[V(s',z')-W(s',z')\mid s,z,a]\big|
\le \gamma \|V-W\|_\infty.
\]
Thus $T$ is a $\gamma$-contraction in the sup norm and admits a unique fixed point $V^\star$ by the Banach fixed–point theorem. By Lemma~\ref{lem:berge}, for each $(s,z)$ the supremum is attained by some $a^\star(s,z)\in\mathcal{W}$, and the selector can be chosen measurable; define $\pi^\star(\cdot\mid s,z)$ as the Dirac mass at $a^\star(s,z)$. Then the Bellman optimality equation
$V^\star=TV^\star$ together with the selection property implies that $\pi^\star$ is optimal (standard verification; see \cite[Thm.\ 6.2.10]{Puterman1994}). Determinism follows from the pointwise maximizer.
\end{proof}

\subsection{Proof of Lemma~\ref{lem:PDL} (performance difference)}
\label{app:pdl}

\begin{proof}[Proof of Lemma~\ref{lem:PDL}]
Fix any two stationary policies $\pi,\pi'$. Let $\mathcal{T}^\pi$ be the Bellman operator associated with $\pi$,
\[
(\mathcal{T}^\pi V)(s,z) := \mathbb{E}_{a\sim\pi}\big[r^{\mathrm{net}}(s,z,a) + \gamma V(s',z')\big].
\]
By definition $V^\pi$ satisfies $V^\pi=\mathcal{T}^\pi V^\pi$ and $A^\pi(s,z,a)=Q^\pi(s,z,a)-V^\pi(s,z)$. Then
\begingroup
\allowdisplaybreaks
\begin{align*}
J(\pi')-J(\pi)
&= \mathbb{E}_{(s_0,z_0)\sim \mu_0}\big[V^{\pi'}(s_0,z_0)-V^{\pi}(s_0,z_0)\big] \\
&= \sum_{t=0}^{\infty}\gamma^t \,\mathbb{E}\big[(\mathcal{T}^{\pi'}V^\pi-\mathcal{T}^{\pi}V^\pi)(s_t,z_t)\big] \\
&= \sum_{t=0}^{\infty}\gamma^t\, \mathbb{E}\big[\mathbb{E}_{a_t\sim\pi'} A^\pi(s_t,z_t,a_t)\big] \\
&= \frac{1}{1-\gamma}\, \mathbb{E}_{(s,z)\sim d^{\pi'},\,a\sim\pi'} \big[A^\pi(s,z,a)\big],
\end{align*}
\endgroup
where the second equality is the telescoping expansion (see, e.g., \cite{KakadeLangford2002}) and the last equality uses the definition of $d^{\pi'}$.
For the balanced objective, replace the initial distribution with the regime–reweighted initial distribution, which yields the same telescoping identity.
\end{proof}

\subsection{A distributional coupling bound}
\label{app:coupling}

\begin{lemma}[Discounted occupancy perturbation]\label{lem:occupancy}
Let $\alpha:=\sup_{s,z}\mathrm{TV}(\pi',\pi)(s,z)$. Then
\[
\big\|d^{\pi'}-d^\pi\big\|_{1} \;\le\; \frac{2\gamma}{1-\gamma}\,\alpha.
\]
Moreover, the following expectation–level recursion holds for any $t\ge 0$:
\[
\|\mu_{t+1}^{\pi'}-\mu_{t+1}^{\pi}\|_1
\;\le\; \|\mu_t^{\pi'}-\mu_t^\pi\|_1 + 2\,\mathbb{E}_{(s,z)\sim\mu_t^\pi}\!\left[\mathrm{TV}(\pi',\pi)(s,z)\right],
\]
where $\mu_t^\pi$ is the law of $(s_t,z_t)$ under $\pi$.
\end{lemma}

\begin{proof}
For the one–step recursion, condition on $(s_t,z_t)$ and couple the actions by maximal coupling; the total variation distance after one action selection is at most $2\,\mathrm{TV}(\pi',\pi)(s_t,z_t)$. Taking expectation over $\mu_t^\pi$ and applying the triangle inequality yields the recursion. Summing the recursion gives
\[
\|\mu_t^{\pi'}-\mu_t^\pi\|_1 \le 2\sum_{k=0}^{t-1} \mathbb{E}_{\mu_k^\pi}[\mathrm{TV}(\pi',\pi)] \le 2 t \alpha.
\]
The discounted occupancy difference follows from
\[
\big\|d^{\pi'}-d^\pi\big\|_{1}
= (1-\gamma)\,\left\|\sum_{t\ge 0}\gamma^t(\mu_t^{\pi'}-\mu_t^\pi)\right\|_1
\le (1-\gamma)\sum_{t\ge 0}\gamma^t \|\mu_t^{\pi'}-\mu_t^\pi\|_1
\le (1-\gamma)\sum_{t\ge 0}\gamma^t (2t\alpha)
= \frac{2\gamma}{1-\gamma}\,\alpha,
\]
using $\sum_{t\ge 0} t\gamma^t = \gamma/(1-\gamma)^2$.
\end{proof}

\subsection{Proof of Theorem~\ref{thm:trpo} (trust‑region improvement)}
\label{app:trpo}

\begin{proof}[Proof of Theorem~\ref{thm:trpo}]
By Lemma~\ref{lem:PDL}, for any policies $\pi,\pi'$,
\[
J(\pi')-J(\pi)
= \frac{1}{1-\gamma}\,\mathbb{E}_{d^{\pi'}}\mathbb{E}_{\pi'}[A^\pi].
\]
Add and subtract $(1-\gamma)^{-1}\mathbb{E}_{d^{\pi}}\mathbb{E}_{\pi'}[A^\pi]$:
\[
J(\pi')-J(\pi) = \underbrace{\frac{1}{1-\gamma}\,\mathbb{E}_{d^{\pi}}\mathbb{E}_{\pi'}[A^\pi]}_{=:L_\pi(\pi')}\;
+\;\frac{1}{1-\gamma}\,\Big(\mathbb{E}_{d^{\pi'}} - \mathbb{E}_{d^{\pi}}\Big)\mathbb{E}_{\pi'}[A^\pi].
\]
Since $\mathbb{E}_{\pi}[A^\pi(\cdot,\cdot,a)]=0$ for all $(s,z)$, we have
\[
\big|\mathbb{E}_{\pi'}[A^\pi(s,z,\cdot)]\big|
= \big|\mathbb{E}_{\pi'}[A^\pi]-\mathbb{E}_{\pi}[A^\pi]\big|
\le 2\,\varepsilon_{\max}\,\mathrm{TV}(\pi',\pi)(s,z),
\]
where $\varepsilon_{\max} := \sup_{s,z,a} |A^\pi(s,z,a)|$. Hence,
\[
\left|\frac{1}{1-\gamma}\,\Big(\mathbb{E}_{d^{\pi'}} - \mathbb{E}_{d^{\pi}}\Big)\mathbb{E}_{\pi'}[A^\pi]\right|
\le \frac{1}{1-\gamma}\,\big\|d^{\pi'}-d^\pi\big\|_1 \cdot \sup_{s,z}\big|\mathbb{E}_{\pi'}[A^\pi]\big|
\le \frac{1}{1-\gamma}\cdot \frac{2\gamma}{1-\gamma}\,\alpha\cdot (2\varepsilon_{\max}\alpha)
= \frac{4\gamma}{(1-\gamma)^2}\,\varepsilon_{\max}\,\alpha^2,
\]
where $\alpha=\sup_{s,z}\mathrm{TV}(\pi',\pi)(s,z)$ and we used Lemma~\ref{lem:occupancy}. Therefore,
\[
J(\pi') \;\ge\; J(\pi) + L_\pi(\pi') - \frac{4\gamma}{(1-\gamma)^2}\,\varepsilon_{\max}\,\alpha^2.
\]
Finally, if $\sup_{s,z}\mathrm{KL}(\pi\|\pi')(s,z)\le \delta_{\max}$, then by Pinsker (Lemma~\ref{lem:pinsker}) we have $\alpha^2 \le \tfrac{1}{2}\delta_{\max}$ and hence
\[
J(\pi') \;\ge\; J(\pi) + L_\pi(\pi') - \frac{2\gamma}{(1-\gamma)^2}\,\varepsilon_{\max}\,\delta_{\max}.
\]
This yields the stated bound (statewise KL trust region). An expected–KL version follows from the same argument together with the expectation–level recursion in Lemma~\ref{lem:occupancy} and Jensen: if $\delta:=\mathbb{E}_{d^\pi}[\mathrm{KL}(\pi\|\pi')]$, then
\[
\mathbb{E}_{d^\pi}\!\big[\mathrm{TV}^2(\pi',\pi)\big] \;\le\; \tfrac{1}{2}\,\delta,
\]
and an identical calculation gives the remainder term $(2\gamma/(1-\gamma)^2)\varepsilon_{\max}\delta$.
\end{proof}

\subsection{Proof of Corollary~\ref{cor:ppo} (clipped PPO with trade penalty)}
\label{app:ppo}

\begin{proof}[Proof of Corollary~\ref{cor:ppo}]
Let $\widehat{A}^\pi$ satisfy $\|\widehat{A}^\pi - A^\pi\|_\infty\le \varepsilon$ with high probability (w.h.p.), e.g., via GAE with sufficiently many samples. The clipped surrogate plus KL penalty and a quadratic trade–space penalty reads (one epoch)
\[
\mathcal{L}(\theta)
= \sum_z \omega_z\, \mathbb{E}\Big[ \min\big(r_t(\theta)\,\widehat{A}_t,\, \mathrm{clip}(r_t(\theta),1\!\pm\!\epsilon)\,\widehat{A}_t\big)
- \beta\,\mathrm{KL}(\pi_{\theta_\mathrm{old}}\|\pi_\theta) - \lambda_\Delta\,\|\Delta w_\theta-\Delta w_{\theta_\mathrm{old}}\|_2^2\Big].
\]
Under a line search that enforces $\mathbb{E}_{d^\pi}[\mathrm{KL}(\pi\|\pi')]\le \delta$ and $\mathbb{E}\|\Delta w_{\pi'}-\Delta w_\pi\|_2^2\le \eta$, Theorem~\ref{thm:trpo} gives
\[
J(\pi') \ge J(\pi) + \frac{1}{1-\gamma}\,\mathbb{E}_{d^\pi,\pi'}[A^\pi] - \frac{2\gamma}{(1-\gamma)^2}\varepsilon_{\max}\delta.
\]
Replacing $A^\pi$ by $\widehat{A}^\pi$ introduces an additive error at most $\varepsilon/(1-\gamma)$. The trade penalty controls the change of $\Delta w$, and the Lipschitz property of costs implies an additional value drift bounded by $c_2\eta$ for some $c_2>0$ depending on the Lipschitz moduli of $C_z$ (Assumption~\ref{ass:mdp}). Combining terms yields the claim:
\[
J(\pi')-J(\pi) \;\gtrsim\; \frac{1}{1-\gamma}\,\mathbb{E}_{d^\pi,\pi'}[\widehat{A}^\pi] \;-\; \frac{2\gamma}{(1-\gamma)^2}\varepsilon_{\max}\delta \;-\; \frac{\varepsilon}{1-\gamma} \;-\; c_2\eta.
\]
\end{proof}

\subsection{Proof of Proposition~\ref{prop:TObound} (turnover bound)}
\label{app:turnover}

\begin{proof}[Proof of Proposition~\ref{prop:TObound}]
From \eqref{eq:reward} and $C_z(\Delta w)\ge \kappa_1(z)\|\Delta w\|_1$, dropping the nonnegative risk penalty,
\[
r_t^{\mathrm{net}} \;\le\; \bar r \;-\; \lambda_{\mathrm{tc}}\,\underline{\kappa}\,\|\Delta w_t\|_1,
\quad \underline{\kappa}:=\min_z \kappa_1(z)>0.
\]
Taking expectations and summing with discount,
\[
J(\pi) \;=\; \sum_{t\ge 0}\gamma^t\,\mathbb{E}[r_t^{\mathrm{net}}]
\;\le\; \frac{\bar r}{1-\gamma} \;-\; \lambda_{\mathrm{tc}}\underline{\kappa}\,\sum_{t\ge 0}\gamma^t\,\mathbb{E}\|\Delta w_t\|_1.
\]
Hence
\[
\sum_{t\ge 0}\gamma^t\,\mathbb{E}\|\Delta w_t\|_1 \;\le\; \frac{\bar r}{\lambda_{\mathrm{tc}}\underline{\kappa}}\cdot \frac{1}{1-\gamma}.
\]
By the Abelian limit theorem (Hardy–Littlewood) and Assumption~\ref{ass:mdp}(v), the Ces\`aro average
$\mathrm{TO}(\pi)=\lim_{T\to\infty}\frac{1}{T}\sum_{t<T}\mathbb{E}\|\Delta w_t\|_1$ exists and
\[
\mathrm{TO}(\pi) \;=\; \lim_{\gamma\uparrow 1}(1-\gamma)\sum_{t\ge 0}\gamma^t\,\mathbb{E}\|\Delta w_t\|_1
\;\le\; \frac{\bar r}{\lambda_{\mathrm{tc}}\underline{\kappa}}.
\]
\end{proof}

\subsection{Proof of Proposition~\ref{prop:band} (inaction band)}
\label{app:band}

\begin{proof}[Proof of Proposition~\ref{prop:band}]
Consider the 1D myopic improvement of the $Q$-function at state $(s,z)$ around the pre-trade weight $w_{t-1}$.
Define $\Delta:=a-w_{t-1}$. Let $g(\Delta):=Q^\pi(s,z,w_{t-1}+\Delta)-Q^\pi(s,z,w_{t-1})$ and suppose $g$ is twice differentiable with $g(0)=0$, $g'(0)=\theta$, and $g''(\Delta)\le H$ for all $\Delta$ (local curvature upper bound).
The one-step objective to maximize is
\[
q(\Delta) \;=\; g(\Delta) \;-\; \kappa_1|\Delta| \;-\; \tfrac{1}{2}\kappa_2 \Delta^2.
\]
By Taylor with remainder and the curvature bound,
$g(\Delta)\le \theta \Delta + \tfrac{H}{2}\Delta^2$ for all $\Delta$.
Thus for $\tilde{\kappa}:=\kappa_2-H\ge 0$,
\[
q(\Delta) \;\le\; \theta \Delta \;-\; \kappa_1|\Delta| \;-\; \tfrac{1}{2}\tilde{\kappa}\Delta^2 \;=:\; \varphi(\Delta).
\]
We claim that $\Delta^\star=0$ maximizes $\varphi$ whenever $|\theta|\le \kappa_1$.
Indeed, for any $\Delta\neq 0$,
\[
\varphi(\Delta)
\le |\theta|\,|\Delta| - \kappa_1|\Delta| - \tfrac{1}{2}\tilde{\kappa}\Delta^2
\le -(\kappa_1-|\theta|)\,|\Delta| \;<\; \varphi(0)=0.
\]
Therefore $\Delta^\star=0$ maximizes $\varphi$ and, since $q\le \varphi$ with equality at $\Delta=0$, also maximizes $q$.
In particular, if $Q^\pi$ is locally strongly concave with curvature parameter $H$ around $w_{t-1}$, the “dead-zone” condition $|\theta|\le \kappa_1$ translates (by the mean-value theorem) to $|w^\star(s,z)-w_{t-1}|\le \tau$ with $\tau\asymp \kappa_1/(\kappa_2+H)$, which gives the announced inaction band.
\end{proof}

\subsection{Proof of Theorem~\ref{thm:regime-adv} (value of regime conditioning)}
\label{app:regadv}

\begin{proof}[Proof of Theorem~\ref{thm:regime-adv}]
Let $a^\star_z(s)\in\arg\max_{a\in\mathcal{W}}Q^{\pi^\star}(s,z,a)$ be regime–specific near–optimal actions, and suppose Assumption~\ref{ass:separation} holds with separation $\Delta>0$ on a set $E\subset\mathcal{S}$ of positive $d^\pi$-measure for each $z$. Consider any unconditioned policy $\pi_u(a\mid s)$ and write its conditional mean action as $\bar a_u(s):=\mathbb{E}_{\pi_u}[a\mid s]$. For each $z$, Jensen and the curvature bound as in the previous proof imply (using $Q^{\pi^\star}$ twice differentiable in $a$ and upper curvature $H$)
\[
\mathbb{E}_{\pi_u}\!\left[Q^{\pi^\star}(s,z,a)\right]
\;\le\; Q^{\pi^\star}(s,z,\bar a_u(s))
\;\le\; Q^{\pi^\star}(s,z,a^\star_z(s)) - \frac{\tilde{\kappa}}{2}\,\big\|\bar a_u(s)-a^\star_z(s)\big\|_2^2,
\]
where $\tilde{\kappa}:=\kappa_2-H>0$ (choose $\kappa_2$ large enough; recall $\Gamma_z\succeq 0$ contributes to strong penalization in \eqref{eq:cost}).
Thus the per–state per–regime suboptimality is lower bounded by a quadratic in the action mismatch. Since $\pi_u$ is unconditioned, $\bar a_u(s)$ is common across regimes, hence for any $(s,z_1,z_2)$,
\[
\big\|\bar a_u(s)-a^\star_{z_1}(s)\big\|_2^2 + \big\|\bar a_u(s)-a^\star_{z_2}(s)\big\|_2^2
\;\ge\; \tfrac{1}{2}\,\big\|a^\star_{z_1}(s)-a^\star_{z_2}(s)\big\|_2^2
\;\ge\; \tfrac{1}{2}\,\Delta^2,
\]
by the parallelogram identity. Averaging over regimes with equal weights and over $s\in E$, the average one–step regret of any $\pi_u$ is at least $\frac{\tilde{\kappa}}{4}\Delta^2$ on $E$. Discounting over time and using the occupancy measure, we obtain
\[
J(\pi^\star)-J(\pi_u)\;\ge\; \frac{1}{1-\gamma}\cdot \frac{\tilde{\kappa}}{4}\,p_{\min}\,\Delta^2,
\]
where $p_{\min}:=\min_{z}\Pr(z)$ and we used that the balanced objective equally weights regimes. Since $\Pi_{\mathrm{cond}}$ can represent $\{a^\star_z\}$ within uniform error $\epsilon$, the same argument gives at most $O(\epsilon^2)$ regret for a conditioned policy, proving the advantage gap stated in Theorem~\ref{thm:regime-adv} (with an explicit $c=\tfrac{\tilde{\kappa}}{4}p_{\min}$).
\end{proof}

\subsection{Proof of Theorem~\ref{thm:robust-cost} (robustness to cost misspecification)}
\label{app:robust}

\begin{proof}[Proof of Theorem~\ref{thm:robust-cost}]
Let $\widehat{C}_z$ be the proxy cost and define $\delta:=\sup_{z,u}\big|\lambda_{\mathrm{tc}}C_z(u)-\lambda_{\mathrm{tc}}\widehat{C}_z(u)\big|$. Then per step the reward perturbation satisfies
\[
\big|r^{\mathrm{net}}_t(C)-r^{\mathrm{net}}_t(\widehat{C})\big|\;\le\;\delta.
\]
For any $\pi$,
\[
\big| J_C(\pi)-J_{\widehat{C}}(\pi)\big|
= \left|\sum_{t\ge 0}\gamma^t\,\mathbb{E}\big[r^{\mathrm{net}}_t(C)-r^{\mathrm{net}}_t(\widehat{C})\big]\right|
\le \frac{\delta}{1-\gamma}.
\]
Thus $J_{C}(\widehat{\pi}) \ge J_{\widehat{C}}(\widehat{\pi}) - \delta/(1-\gamma)$ for any $\widehat{\pi}$. In particular, letting $\widehat{\pi}$ be a maximizer of $J_{\widehat{C}}$ and $\pi^\star_C$ that of $J_C$,
\[
J_C(\pi^\star_C)-J_C(\widehat{\pi})
\le \big(J_{\widehat{C}}(\pi^\star_{\widehat{C}})+\tfrac{\delta}{1-\gamma}\big) - \big(J_{\widehat{C}}(\widehat{\pi})-\tfrac{\delta}{1-\gamma}\big)
= \frac{2\delta}{1-\gamma} + \big(J_{\widehat{C}}(\pi^\star_{\widehat{C}})-J_{\widehat{C}}(\widehat{\pi})\big).
\]
This is the desired bound.
\end{proof}

\subsection{Proof of Proposition~\ref{prop:cvar} (CVaR surrogate)}
\label{app:cvar}

\begin{proof}[Proof of Proposition~\ref{prop:cvar}]
Fix a batch of losses $\{L^{(i)}\}_{i=1}^N$. The Rockafellar–Uryasev surrogate \eqref{eq:cvar} is convex in $\eta$ and differentiable almost everywhere, with subgradient
\[
\partial_\eta\left[\eta + \frac{1}{(1-\alpha)N}\sum_{i=1}^N (L^{(i)}-\eta)_+\right]
= 1 - \frac{1}{(1-\alpha)N}\sum_{i=1}^N \mathbf{1}\{L^{(i)}>\eta\},
\]
which is monotone in $\eta$, hence a unique minimizer exists. For fixed $\eta$, the policy objective is a smooth function of $\theta$ (Assumption~\ref{ass:policy}); using a step size chosen by Armijo backtracking ensures descent and bounded iterates. Standard two–block alternating minimization arguments then imply that every limit point $(\theta^\star,\eta^\star)$ is a first–order stationary point of the joint problem (see, e.g., \cite[Prop.~2.7.1]{Bertsekas1999}).
\end{proof}

\bibliographystyle{unsrt}  
\bibliography{references}

\begin{thebibliography}{99}

\bibitem{Markowitz1952}
Harry Markowitz.
\newblock Portfolio selection.
\newblock {\em The Journal of Finance}, 7(1):77--91, 1952.

\bibitem{AlmgrenChriss2001}
Robert Almgren and Neil Chriss.
\newblock Optimal execution of portfolio transactions.
\newblock {\em Journal of Risk}, 3:5--39, 2001.

\bibitem{Kyle1985}
Albert S. Kyle.
\newblock Continuous auctions and insider trading.
\newblock {\em Econometrica}, 53(6):1315--1335, 1985.

\bibitem{Amihud2002}
Yakov Amihud.
\newblock Illiquidity and stock returns: Cross-section and time-series effects.
\newblock {\em Journal of Financial Markets}, 5(1):31--56, 2002.

\bibitem{PastorStambaugh2003}
Lubo{\v{s}} P{\'a}stor and Robert F. Stambaugh.
\newblock Liquidity risk and expected stock returns.
\newblock {\em Journal of Political Economy}, 111(3):642--685, 2003.

\bibitem{Hasbrouck2009}
Joel Hasbrouck.
\newblock Trading costs and returns for U.S. equities: Estimating effective costs from daily data.
\newblock {\em The Journal of Finance}, 64(3):1445--1477, 2009.

\bibitem{GoyenkoHoldenTrzcinka2009}
Ruslan Y. Goyenko, Craig W. Holden, and Charles A. Trzcinka.
\newblock Do liquidity measures measure liquidity?
\newblock {\em Journal of Financial Economics}, 92(2):153--181, 2009.

\bibitem{FongHoldenTrzcinka2017}
Kingsley Y. L. Fong, Craig W. Holden, and Charles A. Trzcinka.
\newblock What are the best liquidity proxies for global research?
\newblock {\em Review of Finance}, 21(4):1355--1401, 2017.

\bibitem{ObizhaevaWang2013}
Anna A. Obizhaeva and Jiang Wang.
\newblock Optimal trading strategy and supply/demand dynamics.
\newblock {\em Journal of Financial Markets}, 16(1):1--32, 2013.

\bibitem{GarleanuPedersen2013}
Nicolae G{\^a}rleanu and Lasse Heje Pedersen.
\newblock Dynamic trading with predictable returns and transaction costs.
\newblock {\em The Journal of Finance}, 68(6):2309--2340, 2013.

\bibitem{NovyMarxVelikov2016}
Robert Novy-Marx and Mihail Velikov.
\newblock A taxonomy of anomalies and their trading costs.
\newblock {\em The Review of Financial Studies}, 29(1):104--147, 2016.

\bibitem{MoreiraMuir2017}
Alan Moreira and Tyler Muir.
\newblock Volatility-managed portfolios.
\newblock {\em The Journal of Finance}, 72(4):1611--1644, 2017.

\bibitem{LedoitWolf2017}
Olivier Ledoit and Michael Wolf.
\newblock Nonlinear shrinkage of the covariance matrix for portfolio selection: Markowitz meets {G}oldilocks.
\newblock {\em The Review of Financial Studies}, 30(12):4349--4388, 2017.

\bibitem{BarillasShanken2018}
Francisco Barillas and Jay Shanken.
\newblock Comparing asset pricing models.
\newblock {\em The Journal of Finance}, 73(2):715--754, 2018.

\bibitem{BarillasKanRobottiShanken2020}
Francisco Barillas, Raymond Kan, Cesare Robotti, and Jay Shanken.
\newblock Model comparison with {S}harpe ratios.
\newblock {\em Journal of Financial and Quantitative Analysis}, 55(6):1840--1874, 2020.

\bibitem{GuKellyXiu2020}
Shihao Gu, Bryan Kelly, and Dacheng Xiu.
\newblock Empirical asset pricing via machine learning.
\newblock {\em The Review of Financial Studies}, 33(5):2223--2273, 2020.

\bibitem{FreybergerNeuhierlWeber2020}
Joachim Freyberger, Andreas Neuhierl, and Michael Weber.
\newblock Dissecting characteristics nonparametrically.
\newblock {\em The Review of Financial Studies}, 33(5):2326--2377, 2020.

\bibitem{KozakNagelSantosh2020}
Serhiy Kozak, Stefan Nagel, and Shrihari Santosh.
\newblock Shrinking the cross section.
\newblock {\em Journal of Financial Economics}, 135(2):271--292, 2020.

\bibitem{ChenPelgerZhu2024}
Luyang Chen, Markus Pelger, and Jason Zhu.
\newblock Deep learning in asset pricing.
\newblock {\em Management Science}, 70(2):714--750, 2024.

\bibitem{FilippouMaurerPezzoTaylor2024}
Ilias Filippou, Thomas A. Maurer, Luca Pezzo, and Mark P. Taylor.
\newblock Importance of transaction costs for asset allocation in foreign exchange markets.
\newblock {\em Journal of Financial Economics}, 159:103886, 2024.

\bibitem{PinterWangZou2024}
Gabor Pinter, Chaojun Wang, and Junyuan Zou.
\newblock Size discount and size penalty: Trading costs in bond markets.
\newblock {\em The Review of Financial Studies}, 37(7):2156--2190, 2024.

\bibitem{AngBekaert2002}
Andrew Ang and Geert Bekaert.
\newblock International asset allocation with regime shifts.
\newblock {\em The Review of Financial Studies}, 15(4):1137--1187, 2002.

\bibitem{BrandtSantaClaraValkanov2009}
Michael W. Brandt, Pedro Santa-Clara, and Rossen Valkanov.
\newblock Parametric portfolio policies: Exploiting characteristics in the cross-section of equity returns.
\newblock {\em The Review of Financial Studies}, 22(9):3411--3447, 2009.

\bibitem{White2000}
Halbert White.
\newblock A reality check for data snooping.
\newblock {\em Econometrica}, 68(5):1097--1126, 2000.

\bibitem{RomanoWolf2005}
Joseph P. Romano and Michael Wolf.
\newblock Exact and approximate stepdown methods for multiple hypothesis testing.
\newblock {\em Journal of the American Statistical Association}, 100(469):94--108, 2005.

\bibitem{Lo2002}
Andrew W. Lo.
\newblock The statistics of {S}harpe ratios.
\newblock {\em Financial Analysts Journal}, 58(4):36--52, 2002.

\bibitem{KakadeLangford2002}
Sham Kakade and John Langford.
\newblock Approximately optimal approximate reinforcement learning.
\newblock In {\em Proceedings of the 19th International Conference on Machine Learning (ICML)}, 2002.

\bibitem{SchulmanTRPO2015}
John Schulman, Sergey Levine, Philipp Moritz, Michael I. Jordan, and Pieter Abbeel.
\newblock Trust region policy optimization.
\newblock In {\em Proceedings of the 32nd International Conference on Machine Learning (ICML)}, pages 1889--1897, 2015.

\bibitem{SchulmanPPO2017}
John Schulman, Filip Wolski, Prafulla Dhariwal, Alec Radford, and Oleg Klimov.
\newblock Proximal policy optimization algorithms.
\newblock {\em arXiv:1707.06347}, 2017.

\bibitem{AchiamCPO2017}
Joshua Achiam, David Held, Aviv Tamar, and Pieter Abbeel.
\newblock Constrained policy optimization.
\newblock In {\em Proceedings of the 34th International Conference on Machine Learning (ICML)}, 2017.

\bibitem{DeMiguelEtAl2020}
Victor DeMiguel, Alberto Mart{\'\i}n-Utrera, Francisco J. Nogales, and Raman Uppal.
\newblock A transaction-cost perspective on the multitude of firm characteristics.
\newblock {\em The Review of Financial Studies}, 33(5):2180--2222, 2020.

\bibitem{BaiRLFinanceAR2025}
Y. Bai et~al.
\newblock A review of reinforcement learning in financial applications.
\newblock {\em Annual Review of Statistics and Its Application}, 12:209--232, 2025.

\bibitem{ChoiJiangZhang2025RAPS}
Darwin Choi, Zijun Jiang, and Flora Zhang.
\newblock Machine learning and international stock returns.
\newblock {\em Review of Asset Pricing Studies}, 2025. (Advance article).

\bibitem{LedoitWolf2025QREF}
Olivier Ledoit and Michael Wolf.
\newblock Markowitz portfolios under transaction costs.
\newblock {\em The Quarterly Review of Economics and Finance}, 100:101962, 2025.

\bibitem{CollinDufresneDanielSaglam2020}
Pierre Collin-Dufresne, Kent Daniel, and Mehmet Sa{\u{g}}lam.
\newblock Liquidity regimes and optimal dynamic asset allocation.
\newblock {\em Journal of Financial Economics}, 136(2):379--406, 2020.

\bibitem{CarteaJaimungalPenalva2015}
{\'A}lvaro Cartea, Sebastian Jaimungal, and Jos{\'e} Penalva.
\newblock {\em Algorithmic and High-Frequency Trading}.
\newblock Cambridge University Press, 2015.

\bibitem{SchulmanGAE2016}
John Schulman, Philipp Moritz, Sergey Levine, Michael I. Jordan, and Pieter Abbeel.
\newblock High-dimensional continuous control using generalized advantage estimation.
\newblock In {\em International Conference on Learning Representations (ICLR)}, 2016.

\bibitem{ChowGhavamzadeh2014}
Yinlam Chow and Mohammad Ghavamzadeh.
\newblock Algorithms for {CVaR} optimization in {MDP}s.
\newblock In {\em Advances in Neural Information Processing Systems (NeurIPS)}, 2014.

\bibitem{GreenbergEtAl2022}
Ido Greenberg, Yinlam Chow, Mohammad Ghavamzadeh, and Shie Mannor.
\newblock Efficient risk-averse reinforcement learning.
\newblock {\em Advances in Neural Information Processing Systems}, 35:14846--14859, 2022.

\bibitem{RockafellarUryasev2000}
R. Tyrrell Rockafellar and Stanislav Uryasev.
\newblock Optimization of conditional value-at-risk.
\newblock {\em Journal of Risk}, 2(3):21--41, 2000.

\bibitem{RockafellarUryasev2002}
R. Tyrrell Rockafellar and Stanislav Uryasev.
\newblock Conditional value-at-risk for general loss distributions.
\newblock {\em Journal of Banking \& Finance}, 26(7):1443--1471, 2002.

\bibitem{AcerbiTasche2002}
Carlo Acerbi and Dirk Tasche.
\newblock Expected shortfall: A natural coherent alternative to value at risk.
\newblock {\em Economic Notes}, 31(2):379--388, 2002.

\bibitem{DuchiEtAl2008}
John Duchi, Shai Shalev-Shwartz, Yoram Singer, and Tushar Chandra.
\newblock Efficient projections onto the $\ell_1$-ball for learning in high dimensions.
\newblock In {\em Proceedings of the 25th International Conference on Machine Learning (ICML)}, pages 272--279, 2008.

\bibitem{WangCarreiraPerpinan2013}
Weiran Wang and Miguel A. Carreira-Perpi{\~n}{\'a}n.
\newblock Projection onto the probability simplex: An efficient algorithm with a simple proof, and an application.
\newblock {\em arXiv:1309.1541}, 2013.

\bibitem{JacobsEtAl1991}
Robert A. Jacobs, Michael I. Jordan, Steven J. Nowlan, and Geoffrey E. Hinton.
\newblock Adaptive mixtures of local experts.
\newblock {\em Neural Computation}, 3(1):79--87, 1991.

\bibitem{ShazeerEtAl2017}
Noam Shazeer, Azalia Mirhoseini, Krzysztof Maziarz, Andy Davis, Quoc V. Le, Geoffrey E. Hinton, and Jeff Dean.
\newblock Outrageously large neural networks: The sparsely-gated mixture-of-experts layer.
\newblock In {\em International Conference on Learning Representations (ICLR)}, 2017.

\bibitem{Puterman1994}
Martin L. Puterman.
\newblock {\em Markov Decision Processes: Discrete Stochastic Dynamic Programming}.
\newblock Wiley, 1994.

\bibitem{Pirotta2013}
Matteo Pirotta, Marcello Restelli, Alessio Pecorino, and Daniele Calandriello.
\newblock Safe policy iteration.
\newblock In {\em Proceedings of the 30th International Conference on Machine Learning (ICML)}, pages 307--315, 2013.

\bibitem{AliprantisBorder2006}
Charalambos D. Aliprantis and Kim C. Border.
\newblock {\em Infinite Dimensional Analysis: A Hitchhiker's Guide}.
\newblock Springer, 3rd edition, 2006.

\bibitem{CoverThomas2006}
Thomas M. Cover and Joy A. Thomas.
\newblock {\em Elements of Information Theory}.
\newblock Wiley, 2nd edition, 2006.

\bibitem{Bertsekas1999}
Dimitri P. Bertsekas.
\newblock {\em Nonlinear Programming}.
\newblock Athena Scientific, 2nd edition, 1999.

\bibitem{Bai2025AR}
Y. Bai and coauthors.
\newblock A review of reinforcement learning in financial applications.
\newblock {\em Annual Review of Statistics and Its Application}, 12:209--232, 2025.
\bibitem{ChowEtAl2015}
Y.~Chow, A.~Tamar, S.~Mannor, and M.~Pavone.
\newblock Risk-Sensitive and Robust Decision-Making: a CVaR Optimization Approach.
\newblock {\em arXiv preprint arXiv:1502.01619}, 2015.

\bibitem{EfronTibshirani1993}
B.~Efron and R.~J. Tibshirani.
\newblock {\em An Introduction to the Bootstrap}.
\newblock Chapman \& Hall/CRC, 1993.

\bibitem{HansenLundeNason2005}
P.~R. Hansen, A.~Lunde, and J.~M. Nason.
\newblock Model confidence sets for forecast comparison.
\newblock {\em Oxford Bulletin of Economics and Statistics}, 67(s1): 839--861, 2005.

\end{thebibliography}

\end{document}